\RequirePackage{silence}
\newif\ifarxiv
\arxivtrue

\ifarxiv

\documentclass{article}
\usepackage{amsthm}
\usepackage{amssymb}
\usepackage{amsmath}
\usepackage{hyperref}
\usepackage{cleveref}
\usepackage{thm-restate}
\usepackage{caption}
\usepackage{subcaption}
\usepackage[a4paper,total={140mm,222mm}]{geometry}

\newtheorem{theorem}{Theorem}
\newtheorem{lemma}[theorem]{Lemma}

\theoremstyle{definition}
\newtheorem{definition}[theorem]{Definition}

\title{Ski Rental with Uniform Stochastic Signals} 

\usepackage[numbers]{natbib}

\newtheorem{remark}{Remark}

\else

\fi

\usepackage{xspace}
\usepackage{todonotes}
\usepackage{dsfont}
\usepackage{pgfplots}
\pgfplotsset{compat=1.18}

\newcommand{\USSR}{\textsc{Ussr}\xspace}
\newcommand{\OPT}{\textsc{Opt}\xspace}
\newcommand{\cost}{\textit{cost}\xspace}
\newcommand{\costfst}[3]{\cost(#1,#2,#3)\xspace}
\newcommand{\costf}{\costfst{f}{s}{t}}
\newcommand{\ecost}{\textit{Ecost}\xspace}
\newcommand{\ecostft}[2]{\ecost(#1,#2)\xspace}
\newcommand{\ecostf}{\ecostft{f}{t}}
\newcommand{\eratio}{\textit{ECR}\xspace}
\newcommand{\eratioft}[2]{\eratio(#1,#2)\xspace}
\newcommand{\eratiof}{\eratioft{f}{t}}

\ifarxiv

\date{}

\author{
    Tudor Cardas, Lucas Perotin, Bertrand Simon }
\date{
\small{
Université Grenoble Alpes, Grenoble INP, CNRS, INRIA, LIG, Grenoble, France\\
    tudor.cardas@grenoble-inp.org, lucas.perotin@univ-grenoble-alpes.fr, bertrand.simon@cnrs.fr
}}

\fi

\begin{document}

\maketitle

\begin{abstract}

Ski rental is one of the most fundamental online problems, where an agent
repeatedly decides whether it should continue renting for a service, or buy it
at a one-time fixed cost. In classic online analysis, an adversary may freely
select any arbitrary time horizon $t$ after which the instance terminates, hence
leading to strong competitive ratio lower bounds of 2 (deterministic) and
$e/(e-1)$ (randomized).

We consider a novel but natural model aiming at limiting the adversary's power
without degrading its flexibility. A stochastic \emph{signal}, containing no
information apart from its arrival time, is transmitted to the algorithm at a
time drawn from a known distribution scaled to $[0,t]$.  This framework allows us
to account for visible effects of the adversary's choice to the algorithm.
Focusing as a first step on uniform distributions, we show that the best
possible competitive ratios are approximately $1.478$ (deterministic) and $1.438$
(randomized). Surprisingly, simplistic deterministic strategies already achieve a
performance of $1.489$.
\end{abstract}

\newpage

\section{Introduction}
\label{sec:intro}

The ski rental problem is one of the simplest yet fundamental problems studied
under the competitive analysis framework~\cite{borodin2005online}, modeling the generic dilemma between
irreversible one-time investments and recurring payments (rent, subscription, leasing, 
X-as-a-service, ...). We focus on the following continuous version of this
problem: the algorithm rents an item (e.g., skis) at unit rate until a chosen time $x$ after
which it buys it for a unit cost. The adversary freely decides a time horizon
$t$ after which the instance stops. The goal of the algorithm is to minimize its
competitive ratio, i.e., the worst ratio of its cost ($t$ if $t<x$, otherwise
$1+x$) divided by the optimal cost ($\min(1,t)$).

Under this setting, the best ratios achievable by deterministic and randomized
algorithms are respectively $2$ and
$e/(e-1)\approx1.58$~\cite{karlin1988competitive,karlin1994competitive}. The
deterministic solution buys at time $1$ while the randomized algorithm buys
before time $t\leq1$ with probability $\frac{e^t-1}{e-1}$. Given the fundamental
nature of this problem, a vast direction of research focused on extending the
model to target more complex scenarios, e.g., considering multiple
shops~\cite{ai2014multi}, combinatorial
constraints~\cite{zhang2020combinatorial}, multiple
slopes~\cite{lotker2012rent}, repeated instances~\cite{antoniadis2021learning}.

Our objective is not to complexify the model, but rather to empower the algorithm, or
equivalently diminish the adversary's advantage, in order to achieve a less
pessimistic framework, surpassing one of the main issues of competitive
analysis. A related direction is to assume that the horizon time $t$ is drawn
from a distribution which is known, or partially known, by the
algorithm~\cite{kang2025learningaugmentedskirentaldiscrete,madry2011semi}. Such
models assume strong knowledge on the adversary which is then highly restricted.
Our framework is more closely inspired by the recent concept of \emph{learning-augmented
algorithms}, where algorithms have access to an unreliable prediction about the
future instance~\cite{alps_website}. Ski rental has been one of the first
problems studied under this lens due to its fundamental
nature~\cite{kumar2024improvingonlinealgorithmsml}, and is still currently
studied under a more general framework of stochastic
predictions~\cite{cui2026skirentaldistributionalpredictions,kim2026robust}.

Instead of giving the algorithm access to a \emph{prediction}, we assume that
the instance generates a \emph{signal}, detected by the algorithm, which contains
no information except its arrival time $s$. The signal emission is generated
using a publicly known density function  $\mu:[0,1]\rightarrow\mathbb R^+$ to
obtain a density function $\mu_t$ defined on $[0,t]$ as $\mu_t(x)=\mu(x/t)/t$.
The rationale is that some information about the value of $t$ may be perceptible
by the algorithm. However, it is not possible to estimate the distribution of
$t$ based on $s$ as the value of $t$ is chosen adversarially and may have a
fixed distribution. Nevertheless, the algorithm can leverage this information by
adapting its strategy: detecting a signal soon encourages the algorithm to buy
later as it is beneficial for small $t$, and is an unlikely event for large $t$.
This model has been very recently considered for
scheduling~\cite{cosson2026online} and
secretary~\cite{eberle2026secretaryproblemstochasticprecursor} problems.

We focus more precisely on the special case where $\mu$ is the uniform
distribution, and denote the problem as the \emph{Uniform Signal Ski Rental
problem} (\USSR). An algorithm is here entirely defined by a \emph{response
function} $f:\mathbb R^+\to\mathbb R^+ \cup \{ \infty \}$, where $f(s)$ denotes the buying time if
a signal is received at time $s$. After a few structural results on such
response functions, we show that a simple deterministic response function taking
only 2 different values is enough to surpass the classic randomized lower bound
of $1.58$, hence illustrating the power of this model. Our main results are
matching upper and lower bounds for both deterministic and randomized algorithms.
We further show that small deviations of $\mu$ lead to small perturbations of the guarantees.

\begin{theorem}
The optimal competitive ratios for the \USSR problem are $\approx 1.478$ for
deterministic algorithms and $\approx 1.438$ for randomized algorithms.
\end{theorem}

Intuitively, the objective in both settings for the algorithm is to achieve a
\emph{balanced} competitive ratio, i.e., which is the same for the largest
possible interval of adversarial choices $t$. Establishing lower bounds follows
the widely used Yao's principle, designing a randomized adversarial strategy.
Quite surprisingly, such adversaries only focus on the value of $f(\infty)$ and
not the whole function, while optimizing this function is a large part of the
upper bound designs. 
It should further be noted that a simple response function
taking only 3 distinct values is sufficient to achieve a competitive ratio of
$1.489$, very close to the best deterministic algorithm.

\section{Related Work}
\label{sec:related}

As discussed above, the now classic version of ski rental has been solved for
decades, with the optimal deterministic and randomized competitive ratios equal
to $2$ and $e/(e-1)$ respectively. It appeared as a special case of a problem
named \emph{snoopy caching}~\cite{karlin1988competitive,karlin1994competitive},
and formulations with a different adversary were studied as the \emph{rent-to-buy}
problem~\cite{krishnan1995learning}.

A large family of problems has been inspired by this core framework, aiming at
modeling more complex scenarios, while still maintaining the same basic concept
thus algorithm principles. In the \emph{bahncard} problem, a reduction coupon
can be bought in order to decrease the renting
cost~\cite{fleischer2001bahncard}. In the \emph{TCP acknowledgment} problem,
packets arrive over time and the algorithm must balance between sending as few
grouped acknowledgments as possible and minimizing the delay between a packet
arrival and its acknowledgment~\cite{dooly1998tcp}. In the \emph{multi-slope}
ski rental problem, lower rent prices are available once a setup cost is
paid~\cite{lotker2012rent}. This formulation is very close to the \emph{Dynamic
Power Management} problem, in which computing machines have several idle states
available, with associated wake-up costs, and the objective is to minimize the
energy spent on a sparse workload. This problem has been further generalized to
the \emph{multi-shop} ski rental problem, in which it is possible to switch
shops to access other payment options, paying an extra cost~\cite{ai2014multi}.
An orthogonal direction has focused on multi-item ski rental, with various
assumptions on the item payment
constraints~\cite{zhang2020combinatorial,wu2022competitive}.

The aforementioned works focused on classic worst-case competitive analysis with
a powerful adversary able to freely select the worst time horizon for a given
algorithm. A large and diverse line of work in the much wider online analysis
field has been devoted to overcoming the pessimistic nature of this framework by
decreasing the gap between the adversary and algorithm abilities. Regarding ski
rental, a few works considered that the time horizon is drawn from a probability
distribution, which is either partially~\cite{madry2011semi} or
fully~\cite{kang2025learningaugmentedskirentaldiscrete} known. A more general
framework, known as \emph{learning-augmented algorithms}~\cite{alps_website},
assumes that the algorithm has access to an unreliable prediction about the
instance, and the objective is then to outperform classic algorithms if the
prediction is accurate enough, without degrading too much the performance in
case of erroneous predictions. Ski rental was one of the first problems to be
studied under this framework~\cite{kumar2024improvingonlinealgorithmsml}, where
a deterministic algorithm can for instance be at most $3$-competitive and up to
$1.5$-competitive if predictions are perfect. Numerous papers studied
generalizations of the ski rental problem with access to predictions, or
different prediction models. Notably, two recent papers focused on predicting
the probability distribution of the horizon time, using a distance metric
between distributions in the algorithm
guarantees~\cite{cui2026skirentaldistributionalpredictions,kim2026robust}.

Following this concept of unreliable predictions, \cite{benomar2026non} propose
a new model for non-clairvoyant scheduling in which unfinished tasks can signal
the scheduler that they finished an estimated percentage of their computation,
hence guiding the scheduler into which tasks are likely to be the next to
terminate. This concept is at the origin of the stochastic signal model studied
in the present paper, which connects the frameworks of stochastic information
and learning-augmented algorithms. This idea has been applied to the same
non-clairvoyant scheduling problem~\cite{cosson2026online} in which the authors
show that the competitive ratio of the problem $1||\sum C_i$ goes from 2 in the
classic setting to a value in $[1.2;1.752]$ when each job outputs a signal drawn
uniformly at random during its execution. The secretary problem has also been
studied recently under a similar model, where, among other results, a uniform
signal leads to  a competitive ratio improvement from $1/e$ to
$1/2$~\cite{eberle2026secretaryproblemstochasticprecursor}.

\section{Preliminaries}
\label{sec:prelim}

\subsection{Model definition}
\label{sec:model}

As mentioned at the beginning of this paper, we consider the continuous variant of the ski rental problem. An agent decides at every instant whether it continues renting, for a total renting cost equal to the duration elapsed since time $0$, or it buys, adding a definitive unit cost of $1$. An adversary chooses the time $t$, hidden to the algorithm,  at which the instance stops. The optimal clairvoyant cost is then equal to $\OPT(t)=\min(1,t)$, while an online algorithm buying at time $x$ pays $t$ if $x>t$ and $1+x$ if $x\leq t$. 
The best competitive ratio (maximum for $t$ of the algorithm cost divided by $\OPT(t)$) achievable by a deterministic online algorithm is then $2$, buying at time 1.

In the Uniform Signal Ski Rental problem (\USSR), the algorithm also has access to an information-empty signal $s$ which is drawn uniformly at random in the interval $[0,t]$. Receiving the signal gives no direct information on the value of $t$, but the algorithm can still benefit from it.

An algorithm is then characterized by a response function $f:[0,\infty) \rightarrow \mathbb R^+ \cup\{\infty\}$. If a signal is received at time $s$, then the algorithm buys at time $f(s)$. By definition, as it is impossible to buy retroactively, we then must have $f(s)\geq s$ until a fallback time $s_F\in\mathbb R^+\cup\{\infty\}$ for which $s>s_F\Rightarrow f(s)=s_F$. This means that if no signal is received before time $s_F$, then the algorithm directly buys at time $s_F$. We denote this required property on $f$ as saying that $f$ is \emph{legal}. As $s_F$ depends on $f$, we use the notation $s_F(f)$ if the context is not clear.

\begin{definition}
A function $f:[0,\infty)\to \mathbb R^+\cup\{\infty\}$ is \emph{legal} if there exists $s_F\in\mathbb R^+ \cup \{ \infty \}$ such that $f(x)\geq x$ for $x\leq s_F$ and $f(x)=s_F$ for $x>s_F$.
\end{definition}

The cost of the algorithm defined by $f$ for fixed values of $s$ and $t$ is noted $\costf$. When $s$ is drawn uniformly at random in $[0,t]$, the expected cost of this is denoted by $\ecostf = E_{s\sim U([0,t])}(\costf)$. We also denote the ratio of this expected cost divided by the optimal cost: $\eratiof = \ecostf /\min(1,t) $.

The objective is to design algorithms, or equivalently response functions $f$, minimizing the expected competitive ratio:
\begin{align*}
CR(f) = \sup_{t>0} E_s(\costf/\min(1,t)).
\end{align*}

We first focus on deterministic response functions, where the only source of
randomness may come from the arrival time $s$ of the signal. In particular, the
fallback time $s_F$ must be fixed. We then extend the scope to fully randomized
algorithms.

\subsection{First results and characterizations}
\label{sec:firstres}

We start this section by a warm-up remark giving a first upper bound on $s_F$ for relevant functions, using the classic deterministic ski rental algorithm.

\begin{remark}
\label{rmk:sfcr}
For any response function $f$ of the \USSR problem leading to a competitive ratio at most 2, the fallback time $s_F$ is at most 1.
\end{remark}

\begin{proof}
The classic 2-competitive online algorithm corresponds to the constant function $f(s)=1$, buying at time 1 regardless of the signal.

Consider a function $f$ such that $s_F>1$ is finite. For $t>1$, the cost incurred by $f$ is at least $1+s_F$ if $s>s_F$, so the expected cost is at least $(1+s_F)\frac{t-s_F}{t}>2$ for $t$ large enough, while the optimal cost equals 1. If $s_F$ is infinite, the cost incurred by $f$ is at least the signal time $s$, so $t/2$ in expectation.
\end{proof}

This remark allows us to understand the need to obtain better upper bounds on the competitive ratio of the \USSR problem, in order to get better constraints on the structure of optimal response functions. Hence, we now give a simple response function having a surprisingly low competitive ratio smaller than $1.5$. This ensures that better response functions must satisfy $s_F<0.5$. We do not detail here how this function was obtained as it is an approximation of the optimal 3-step function introduced in the next section, whose optimality proof relies on further characterizations.

\begin{lemma}
\label{lem:cr1.5}
There exists a response function $f$ of the \USSR problem composed of 3 steps for which $CR(f)<1.49$.
\end{lemma}

\begin{proof}
Consider $f$ equal to $\infty$ on $[0,s_R)$, equal to $1$ on $[s_R,s_F]$ and to $s_F$ on $(s_F,\infty)$, with $s_R  =0.117$ and $s_F=0.372$.
For $t \geq 1$, the optimal cost is 1 and the algorithm pays $t$ if $s<s_R$, $2$ if $s\in[s_R,s_F]$ and $1+s_F$ if $s>s_F$, we therefore have:
    \begin{align*}
        \eratiof &= \frac{s_R}{t} \cdot t + \frac{s_F-s_R}{t} \cdot 2 + \frac{t-s_F}{t} \cdot (1+s_F)  \\
        &= 1+ s_R + s_F + \frac{s_F-2s_R-s_F^2}{t} 
        = 1.489 - \frac{0.000384}{t} \leq 1.489.
    \end{align*}

For $t\in[s_F,1)$, the optimal cost is $t$ and we have:
\begin{align*}
\eratiof &= \frac1t\left(\frac {s_F} t \cdot t + \frac{t-s_F}{t}(1+s_F)\right) \\
         &= \frac{2s_F+1}{t} - \frac{s_F(1+s_F)}{t^2} = \frac{1.744}{t} - \frac{0.510384}{t^2} \leq \frac{1.744^2}{4\cdot0.510384} <1.49.
\end{align*}

For $t<s_F$, the algorithm never buys so we have $\eratiof=1$.
\end{proof}

\section{Optimal simple step functions (few intervals)}
\label{sec:step}

In this section, we restrict the study to simple step response functions $f$ composed of few steps. Despite the simplicity of the functions, the optimality proof requires lengthy case analysis and is deferred to the appendix.

We first compute the best competitive ratio that can be achieved for 2 steps. We can see on Figure~\ref{fig:2_interval} that this solution balances the cost obtained for $t<1$ and for $t\geq 1$, where the peak is achieved for $t=1$.

\begin{lemma}
\label{lem:2step}
The best competitive ratio achieved for a 2-step response function $f$ for the \USSR problem is $\approx1.55$, buying at time $1$ for $s<s_F\approx0.335$ and $s_F$ otherwise.
\end{lemma}

More formally, the optimal value for $s_F$ is the solution to the equation $8 x^2 + 4 x^3 - 4 x^4 = 1$ belonging to $[0,1]$, and the competitive ratio is $1+2s_F-s_F^2$.

\newcommand{\SetAxisLabels}[1]{\ifnum#1=1
        xlabel={Signal time (s)}, ylabel={Buying time ($f(s)$)}\else
        xlabel={Adversarial horizon (t)}, ylabel={Expected ratio (${ECR}(f,t)$)}\fi
}

\newcommand{\thickness}{ultra thick}

\newcommand{\PlotPiecewise}[8]{\addplot[domain=#1:#2, \thickness, #8] {#5};
    \addplot[domain=#2:#3, \thickness, #8] {#6};
    \addplot[domain=#3:#4, \thickness, #8] {#7};

    \begingroup
	\pgfmathdeclarefunction{x}{0}{\pgfmathparse{#2}}\pgfmathparse{#5}\global\let\yB\pgfmathresult
    \endgroup
    
    \begingroup
	\pgfmathdeclarefunction{x}{0}{\pgfmathparse{#2}}\pgfmathparse{#6}\global\let\yC\pgfmathresult
    \endgroup

    \begingroup
	\pgfmathdeclarefunction{x}{0}{\pgfmathparse{#3}}\pgfmathparse{#6}\global\let\yD\pgfmathresult
    \endgroup
    \begingroup
	\pgfmathdeclarefunction{x}{0}{\pgfmathparse{#3}}\pgfmathparse{#7}\global\let\yE\pgfmathresult
    \endgroup

    \draw[dashed, very thick, #8] (axis cs:#2,{\yC}) -- (axis cs:#2,{\yB});
    \draw[dashed, very thick, #8] (axis cs:#3,{\yD}) -- (axis cs:#3,{\yE});
}

\newcommand{\PlotPiecewisel}[7]{\PlotPiecewise{#1}{#2}{#3}{#4}{#5}{#6}{#7}{blue}
}

\newcommand{\PlotPiecewiser}[7]{\PlotPiecewise{#1}{#2}{#3}{#4}{#5}{#6}{#7}{red}
}

\newcommand{\AxisTick}[4]{\ifstrequal{#1}{x}{
        \draw[dotted, very thick] (axis cs:#2,\pgfkeysvalueof{/pgfplots/ymin}) -- (axis cs:#2,#3) node[at start, below] {#4};
    }{
        \draw[dotted, very thick] (axis cs:0,#2) -- (axis cs:#3,#2);
        \node[left, xshift=-7pt, rotate=90, anchor=center] at (axis cs:0,#2) {#4};
    }}

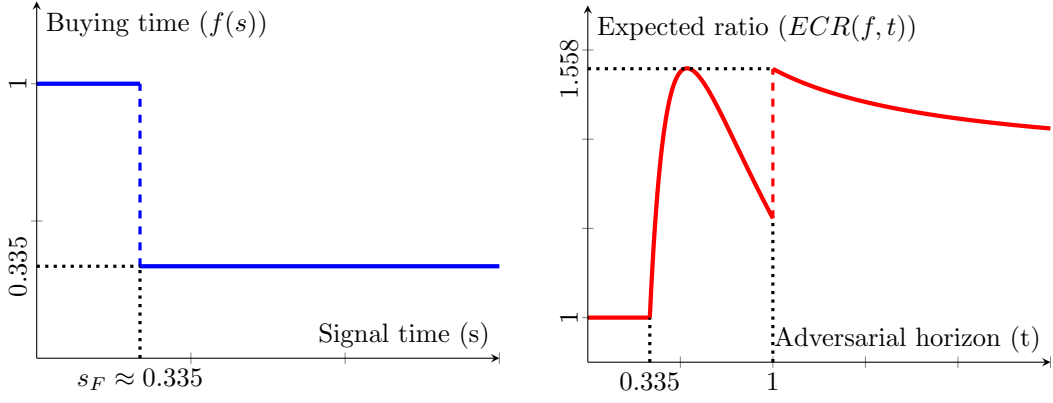
\begin{figure}[htbp]
\centering
\begin{tikzpicture}
    \begin{axis}[
	axis lines=middle,
        \SetAxisLabels{1}, 
		clip=false,
		xticklabel=\empty,
        yticklabel=\empty,
        xmin=0, xmax=1.5, ymin=0, ymax=1.3,
        samples=100,
		width=.55\textwidth,
		height=.45\textwidth
    ]

        \PlotPiecewisel{0}{0.335}{1}{1.5}{1}{0.335}{0.335}
        
        \AxisTick{x}{0.335}{0.335}{$s_F\approx0.335$} 
        \AxisTick{y}{0.335}{0.335}{$0.335$}
        \AxisTick{y}{1}{0}{1}
    \end{axis}
\end{tikzpicture}
\hfill
\begin{tikzpicture}
    \begin{axis}[
	axis lines=middle,
        \SetAxisLabels{2}, 
		clip=false,
		xticklabel=\empty,
        yticklabel=\empty,
        xmin=0, xmax=2.5, ymin=0.9, ymax=1.7,
        samples=100,
		width=.55\textwidth,
		height=.45\textwidth
    ]
       \PlotPiecewiser{0}{0.335}{1}{2.5}{1}{(1+2*0.335)/x-(0.335*1.335)/x/x}{1+0.335+0.335*(1-0.335)/x}
        
        \AxisTick{y}{1.558}{1}{1.558}
        \AxisTick{y}{1}{0}{1}
		\AxisTick{x}{0.335}{1}{0.335}
		\AxisTick{x}{1}{1.22}{1}
    \end{axis}
\end{tikzpicture}
\caption{2-Interval optimal response function (left) and expected ratio (right).}
    \label{fig:2_interval}
\end{figure}

We now compute the best competitive ratio that can be achieved for 2 steps, plus a renting step for signals lower than a given time $s_R$. The optimal solution is close to the one presented in Lemma~\ref{lem:cr1.5}, but the optimality proof is more technical and allows us to understand the function objective: balancing the cost obtained for $t<1$, and for $t\geq1$ for which the expected cost is constant, thanks to the renting step. Indeed, it can be shown that achieving a constant competitive ratio for $t\geq 1$ is equivalent to achieve $\int_{s_R}^{s_F} f=s_R+s_F^2$, which is a far less restrictive constraint with $s_R>0$. This result, in addition to improve the understanding of the problem on simple functions, also shows that such a function achieves a competitive ratio less than $0.012$ away from the optimal.

\begin{lemma}
\label{lem:3step}
The best competitive ratio achieved for a 3-step response function $f$ including a first step at $+\infty$ for the \USSR problem is $\approx1.489$, renting for $s<s_R\approx 0.117$, buying at time $1$ for $s_R<s<s_F\approx0.372$ and at $s_F$ otherwise.
\end{lemma}

More formally, $s_F$ is the solution in $[0,1]$ to the equation $2x^2(3+2x-x^2)=1$, $s_R=\frac12(s_F-s_F^2)$ and the competitive ratio equals $1+s_R+s_F$.

\begin{figure}[htbp]
    \centering
\begin{tikzpicture}
    \begin{axis}[
	axis lines=middle,
        \SetAxisLabels{1}, 
		clip=false,
		xticklabel=\empty,
        yticklabel=\empty,
        xmin=0, xmax=.7, ymin=0, ymax=1.45,
        samples=100,
		width=.55\textwidth,
		height=.45\textwidth
    ]

        \PlotPiecewisel{0}{0.117}{0.372}{.7}{1.25}{1}{0.372}
        
        \AxisTick{x}{0.372}{0.372}{\raisebox{0pt}{$s_F\approx0.372$}} 
\AxisTick{x}{0.117}{1}{$s_R\approx0.117$} 
        \AxisTick{y}{0.372}{0.372}{$0.372$}
        \AxisTick{y}{1}{0.117}{1}
        \AxisTick{y}{1.25}{0}{$\infty$}
    \end{axis}
\end{tikzpicture}
\hfill
\begin{tikzpicture}
    \begin{axis}[
	axis lines=middle,
        \SetAxisLabels{2}, 
		clip=false,
		xticklabel=\empty,
        yticklabel=\empty,
        xmin=0, xmax=2.5, ymin=0.9, ymax=1.7,
        samples=100,
		width=.55\textwidth,
		height=.45\textwidth
    ]
       \PlotPiecewiser{0}{0.372}{1}{2.5}{1}{(1+2*0.372)/x-(0.372*1.372)/x/x}{1+0.372+0.117+(0.372*(1-0.372)-2*0.117)/x}

        \AxisTick{y}{1.489}{1}{1.489}
        \AxisTick{y}{1}{0}{1}
		\AxisTick{x}{0.372}{1}{0.372}
		\AxisTick{x}{1}{1.22}{1}
    \end{axis}
\end{tikzpicture}
\caption{3-Interval optimal response function, including a renting interval (left) and expected ratio (right).}
    \label{fig:3_interval}
\end{figure}
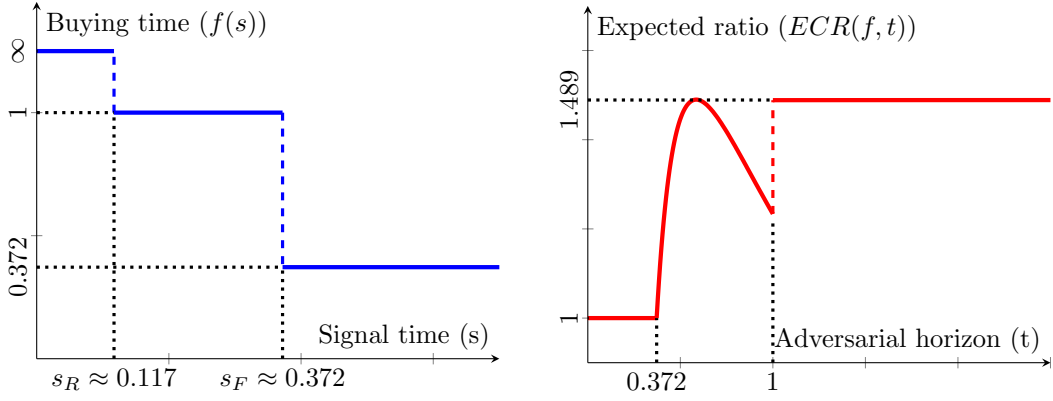

\section{The optimal deterministic solution to the {USSR} problem}

\label{sec:conti}

We now consider general deterministic response functions. We first depict a
function achieving a competitive ratio equal to $C^*$ as defined below, before
using randomized adversaries to prove the optimality of this value.

\begin{theorem}
\label{th:ULB}
The optimal competitive ratio for the \USSR problem equals $C^*:=\frac12+\frac{\sqrt{1+\ln(2e/3)^2}}{2\ln(2e/3)} \approx 1.478$.
\end{theorem}

\subsection{Upper Bound}

The response function $f$ described further below is designed based on
intuitions gained in the previous section. We briefly expose here how it was
obtained. Similarly to the 3-step response function, $f$ equals $\infty$ up to a
value $s_R$, equals a value $s_F$ after $s_F$, and is decreasing and belongs to
$[s^*,1]$ on $[s_R,s_F)$, for some value $s^*>s_F$.

The objective of this function is to minimize the maximum expected competitive
ratio $\eratiof$ for all $t$. This can be decomposed into three properties:
\begin{itemize}
\item for $t \in [s_F,s^*)$, the ratio value is increasing then decreasing
regardless of $f$ over $[s_R,s_F]$, so $s_F$ is chosen to minimize the peak, and $s^*$ to match
the peak.
\item for $t\in [s^*,1]$, the function $f$ over $[s_R,s_F]$ is designed to achieve a
constant competitive ratio.
\item for $t>1$, the competitive ratio should be constant as well,
which is equivalent to verifying $\int_{s_R}^{s_F} f=s_R+s_F^2$, in which case the
competitive ratio equals $1+s_R+s_F$.
\end{itemize}

The first constraint implies a first lower bound on the competitive ratio,
formalized in the next section as Lemma~\ref{lem:LB1}, which is decreasing with
$s_F$. The last constraint, on the opposite, forces $s_R$ and $s_F$ to be small
in order to obtain a good competitive ratio. And for a fixed $s_F$, the
value of $s_R$ is constrained by the value of the integral of $f$ and its shape
forced by the second constraint. The solution was obtained as a compromise
between these bounds.

The optimal response function is defined by the following values:
\begin{align*}
s^* &= \ln \left(\frac{2e}{3} \right) \approx0.595 
&s_F &= \frac{\sqrt{1+{s^*}^2}-(1-s^*)}{2} \approx 0.379 \\
C^* &= \frac 12 + \frac{\sqrt{1+{s^*}^2}}{2s^*} \approx 1.478
&s_R &= C^*-1-s_F \approx0.099.
\end{align*}
On $[s_R,s_F)$, $f$ equals the inverse of the decreasing function $h:[s^*,1]\to[s_R,s_F]$:
$$h(t) = -3C^* e^{t-1} + 2C^*t + 2C^* - (1+s_F).$$

Formally, $f$ can be expressed using the second branch $W_{-1}$ of the Lambert function:
\begin{align*}
f(s) &= \frac{s + 1 + s_F}{2C^*} - 1 - W_{-1}\left(-\frac{3}{2}\exp\left(\frac{s + 1 + s_F}{2C^*} - 2\right)\right) \in\mathbb R \text{~~ for } s \in [s_R,s_F).
\end{align*}

See Figure~\ref{fig:conti} for an illustration, and the appendix for the detailed proof of the competitive ratio of this response function.

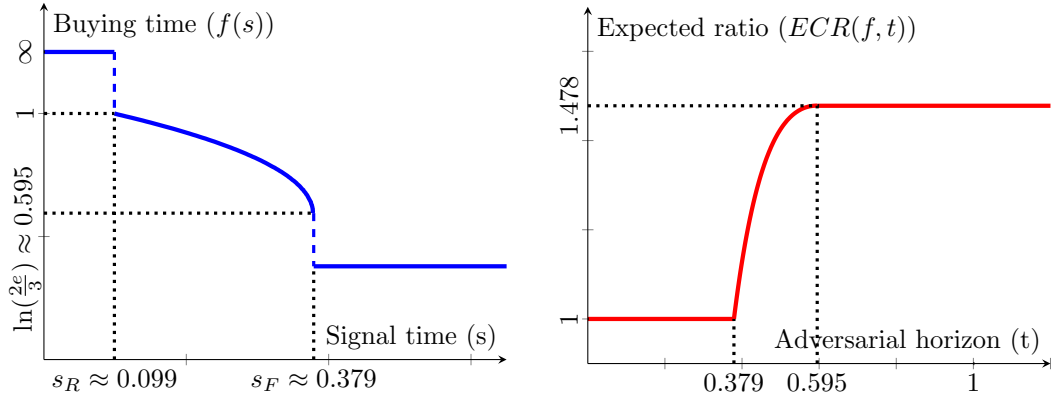
\begin{figure}[htbp]
    \centering
\begin{tikzpicture}
    \begin{axis}[
	axis lines=middle,
        \SetAxisLabels{1}, 
		clip=false,
		xticklabel=\empty,
        yticklabel=\empty,
        xmin=0, xmax=.65, ymin=0, ymax=1.45,
        samples=100,
		width=.55\textwidth,
		height=.45\textwidth
    ]

    	\addplot[domain=0:0.099, \thickness, blue] {1.25};
    	\addplot[domain=0.595:1, \thickness, blue] (-3*1.478*exp(x-1)+2*1.478*x+2*1.478-1-0.379,x);
    	\addplot[domain=0.379:.65, \thickness, blue] {0.379};
    \draw[dashed, very thick, blue] (axis cs:0.099,1.25) -- (axis cs:0.099,{1});
    \draw[dashed, very thick, blue] (axis cs:0.379,{0.595}) -- (axis cs:0.379,{0.379});

        \AxisTick{x}{0.379}{0.379}{\raisebox{0pt}{$s_F\approx0.379$}} 
\AxisTick{x}{0.099}{1}{$s_R\approx0.099$} 
        \AxisTick{y}{0.595}{0.379}{$\ln(\frac{2e}3)\approx0.595$~~~~~~~~~}
        \AxisTick{y}{1}{0.099}{1}
        \AxisTick{y}{1.25}{0}{$\infty$}
    \end{axis}
\end{tikzpicture}
\hfill
\begin{tikzpicture}
    \begin{axis}[
	axis lines=middle,
        \SetAxisLabels{2}, 
		clip=false,
		xticklabel=\empty,
        yticklabel=\empty,
        xmin=0, xmax=1.2, ymin=0.9, ymax=1.7,
        samples=100,
		width=.55\textwidth,
		height=.45\textwidth
    ]
       \PlotPiecewiser{0}{0.379}{0.595}{1.2}{1}{(1+2*0.379)/x-(0.379*1.379)/x/x}{1.478}
        
        \AxisTick{y}{1.478}{0.6}{1.478}
        \AxisTick{y}{1}{0}{1}
		\AxisTick{x}{0.379}{1}{0.379}
		\AxisTick{x}{0.595}{1.478}{\raisebox{0pt}{0.595}}
		\AxisTick{x}{1}{.9}{1}
    \end{axis}
\end{tikzpicture}

\caption{Optimal response function  (left) and expected ratio (right).}
    \label{fig:conti}
\end{figure}

\subsection{Lower Bound}

We establish two independent lower bounds on the competitive ratio of a given
response function $f$ depending on the value of $s_F$. The first one,
Lemma~\ref{lem:LB1}, as mentioned in the previous section, penalizes small
values of $s_F$ by exhibiting a stopping time with a high cost. The second one,
detailed in the main proof of Theorem~\ref{th:det-lower}, uses a stochastic
adversary targeting medium values of $s_F$. For such response functions, the
expected cost on this randomized adversary will be high, which lower bounds
their competitive ratio. The resort to this randomized adversary allows us to
provide a generic proof independent of the actual response function, which is
surprising. It can be explained by a usage similar to Yao's principle, with an
analogy between the signal randomness and a pure randomized algorithm
restricted to a single value of $s_F$. Note that we already determined that
large values of $s_F$ lead to large competitive ratios.

\begin{lemma}
\label{lem:LB1}
For any legal response function $f$ with $0<s_F<1/\sqrt 2$, we have
$CR(f)\geq 1+\frac{1}{4s_F(1+s_F)}$.
\end{lemma}

\begin{proof}
Let $t\in[s_F,1]$. With probability $(t-s_F)/t$, the signal arrives
after $s_F$, so the algorithm buys at $s_F$ and pays $1+s_F$.
Otherwise its cost is at least $t$, since buying costs at least $1\geq t$.
Therefore,
$$
\eratiof\geq
\frac1t\left(\frac{s_F}{t}\cdot t+
\frac{t-s_F}{t}\cdot(1+s_F)\right)
=\frac{1+2s_F}{t}-\frac{s_F(1+s_F)}{t^2}.
$$
The right-hand side is maximized at
$t^*=\frac{2s_F(1+s_F)}{1+2s_F}$, which belongs to $[s_F,1]$ when
$s_F<1/\sqrt 2$. Evaluating at $t^*$ gives
$$
CR(f)\geq\eratioft{f}{t^*}
\geq\frac{(1+2s_F)^2}{4s_F(1+s_F)}
=1+\frac{1}{4s_F(1+s_F)}.
$$
\end{proof}

We are now ready to prove the main lower bound.

\begin{theorem}
\label{th:det-lower}
Every deterministic algorithm for the \USSR problem has competitive ratio
at least $C^*$.
\end{theorem}

\begin{proof}
Fix a legal response function $f$. We may assume $0<s_F<1/2$ as the competitive
ratio would be larger than $1.5$ otherwise, see Remark~\ref{rmk:sfcr}.

We consider a variant of the expected cost achieved by $f$ by scaling it by a factor $t$, which allows the use of expectation over a randomized adversary	
$$
\forall t>0, ~~C_f(t)=t\,\ecostf=\int_0^t\costf\,ds \leq CR(f)\cdot t\OPT(t).
$$
For $t\geq s_F$, separating signals before and after the fallback time $s_F$ gives
\begin{equation*}
C_f(t)=(t-s_F)(1+s_F)+\int_0^{s_F}\costf\,ds.
\end{equation*}

We now design a stochastic adversary depicted by a random variable $T$ representing the time horizon $t$. One objective is to ensure that, for all $s<s_F$, we have $E_T(\costfst f s T)\geq1$.
Let $s^*:=\ln(2e/3)\approx0.595$ and define $T$ as follows:
\begin{align*}
T &= \begin{cases}
s^* & \text{with probability } s^* \\
s^* + Y & \text{with probability } 1-s^*  ~~~~~ (\text{with }Y\sim \mathit{Exp}(1), \text{ the exponential distribution}).
\end{cases}
\end{align*}

Therefore, for $t>s^*$, the density of $T$ equals $(1-s^*)e^{s^*-t}$.
We then have $T\geq s^*>s_F$ and:
$$
\mathbb E[T]=s^*\cdot s^*+(1-s^*)(s^*+1)=1.
$$

By definition of $C_f$, we have:
$$ CR(f) \geq {E[C_f(T)]}~/~{E[T\OPT(T)]}.$$

We therefore focus on $E[C_f(T)]$. As $T\geq s_F$ for all valuations, we have:
$$E[C_f(T)] = (E[T]-s_F)(1+s_F)+\int_0^{s_F}E_T[\cost(f,s,T)]\,ds.$$

Consider the response function $g(s)=\infty$ on $[0,s_F]$ and $s_F$ otherwise,
so $\cost(g,s,T)=T$ for all $s<s_F$. We now show that the expected cost
associated to $f$ on the randomized adversary $T$ is higher than the one of $g$.
We define for $0<s<s_F$:
$$\Delta(f(s)) := E_T(\cost(f,s,T))-E_T(\cost(g,s,T)).$$

If $f(s)>T$, then $f$ and $g$ lead to the same cost. Otherwise, $f$ pays
$1+f(s)$ and $g$ pays $T$. Therefore, for  $x\in\mathbb R^+$, we have:
\begin{align*}
\Delta(x)
&=\mathbb E[(1+x-T)\mathds 1_{\{T\geq x\}}]
=(1+x)\mathbb P(T\geq x)-\mathbb E[T\mathds 1_{\{T\geq x\}}].
\end{align*}

If $x\leq s^*$, then $T\geq x$, so $\Delta(x)=1+x-E[T]=x\geq0$.

If $x>s^*$, direct integration gives:
$$
\mathbb P(T\geq x)=(1-s^*)e^{s^*-x},
\qquad
\mathbb E[T\mathds 1_{\{T\geq x\}}]
=(x+1)(1-s^*)e^{s^*-x},
$$
and hence $\Delta(x)=0$.
Thus $\Delta(x)\geq0$ for every purchase time.

 Consequently,
$$
\mathbb E_T[\costfst{f}{s}{T}]
=\mathbb E_T[\costfst g s T]+\Delta(f(s))\geq1.
$$

We therefore obtain:
\begin{equation}
\label{eq:det-average-cost}
\mathbb E[C_f(T)]
\geq(1-s_F)(1+s_F)+\int_0^{s_F}1\,ds
=1+s_F-s_F^2.
\end{equation}

It remains to compute $E[T\OPT(T)]$, using $3e^{s^*-1}=2$.
\begin{align*}
\mathbb E[T\OPT(T)]
&= E(\min(T,T^2))\\
&={s^*}^2\cdot s^* + \int_{s^*}^{1} t^2(1-s^*)e^{s^*-t}\,dt + \int_1^\infty (1-s^*)te^{s^*-t}\,dt \\
&={s^*}^3 + (1-s^*)[(t^2+2t+2)e^{s^*-t}]_1^{s^*} + (1-s^*)[(t+1) e^{s^*-t}]_\infty^1 \\
&= {s^*}^3 + (1-s^*) \left( {s^*}^2+2s^*+2 - 5e^{s^*-1} +2e^{s^*-1} \right)\\
&= {s^*}^3 + (1-s^*) \cdot s^*(s^*+2)
=s^*(2-s^*).
\end{align*}
Therefore, we obtain the lower bound:
$$CR(f) \geq \frac{\mathbb E(C_f(T))}{\mathbb E[T\OPT(T)]}\geq \frac{1+s_F-s_F^2}{s^*(2-s^*)}.$$
Combining it with Lemma~\ref{lem:LB1} and the initial remark regarding
$s_F<1/2$, we get:
$$
CR(f)\geq\max\left\{
1+\frac{1}{4s_F(1+s_F)},\;
\frac{1+s_F-s_F^2}{s^*(2-s^*)}
\right\}.
$$
On $(0,1/2)$, the first bound is decreasing and the second is increasing.
They are both equal to $C^*=\frac12 + \frac{\sqrt{1+{s^*}^2}}{2s^*}$ at $
s_F=\frac12({\sqrt{1+{s^*}^2}+s^*-1})\approx0.379.
$

\end{proof}

The probability distribution of $T$ has been computed by making buying times
irrelevant after some cutoff $u$. Buying at time $x$ adds $1+x-T$ over
the cost of renting assuming $T\geq x$. An exponential tail of rate
$1$ makes the expected remaining time equal to the buying cost $1$,
so this additional cost averages to zero for every $x>u$.

$T$ is then defined by a point probability $p$ at $u$ followed by an exponential
distribution of rate 1. Always renting  costs $\mathbb E[T]=u+1-p$, whereas
buying immediately costs $1$ and buying before $u$ is suboptimal. Balancing
these costs implies $p=u$.

For $u\in[s_F,1]$, this construction always gives the same numerator in the
competitive ratio lower bound, $1+s_F-s_F^2$. The value of $u$ is then
determined to $\ln(2e/3)$ as it minimizes the denominator, regardless of $s_F$.

\section{The optimal randomized solution to the USSR problem} 
\label{sec:rando} 
 
In this section, we allow the algorithm to use randomness, in addition to the
random signal. We characterize randomized algorithms as a deterministic response
function depending on the valuation of a random variable $\Omega$, with possibly
variable fallback times.
$$ 
CR(\mathcal A)=\sup_{t>0} 
\frac{\mathbb E_\Omega[\ecostft{f_\Omega}{t}]}{\OPT(t)}. 
$$ 
Randomization improves the optimal ratio from $C^*\approx1.478$ to approximately
$1.438$. 

To prove the lower bound, we use a technique similar to the
deterministic lower bound using a randomized adversary. The difference is that
there is no randomized counterpart for Lemma~\ref{lem:LB1}, targeting small
values of $s_F$. Hence, the adversary is required to target a larger interval of
fallback times, decreasing the achieved lower bound.

We first describe the optimal solution, and defer the proof of its expected
competitive ratio to the appendix. This algorithm, once the random bits are
evaluated, follows the same structure as the 3-step function of
Lemma~\ref{lem:3step}, with a renting step until $s_R$ and a constant value
$b_1$ before a fallback $s_F$. As previously, the competitive ratio for large
$t$ approaches $1+\mathbb E(s_R+s_F)$, hence the incentive to minimize these
values. Simultaneously, low time horizons prevent $s_F$ from being too low. Ensuring
that the expected competitive ratio is constant for all time horizons, which intuitively allows us to minimize the peak,  gives
further constraints on the distributions of $b_1$ and $s_F$ and the expected values of $b_1$, $s_F$ and $s_R$.

\begin{theorem}
\label{th:rand-upper}
There exists a randomized algorithm for the \USSR problem with competitive
ratio
$
C_{\mathrm{rand}}
=
\frac{4e^2}{4e^2-9}
\approx1.4378.
$
\end{theorem}

The randomized response function is defined as follows. Set
$
s^*=\ln(2e/3).
$
Draw the fallback $s_F$ in $[0,s^*]$ with density
$
2(C_{\mathrm{rand}}-1)e^{2x}
$, and set:
$$
s_R = \frac{C_{\mathrm{rand}}(1-s^*)-1/2}
     {C_{\mathrm{rand}}s^*-1/2} s_F
\approx0.234 s_F.
$$

Independently of $s_F$, draw a middle buying time $b_1$ in $[s^*,1]$
with density:
$
\frac{2(e^{x-s^*}-1)}{2s^*-1}
$.

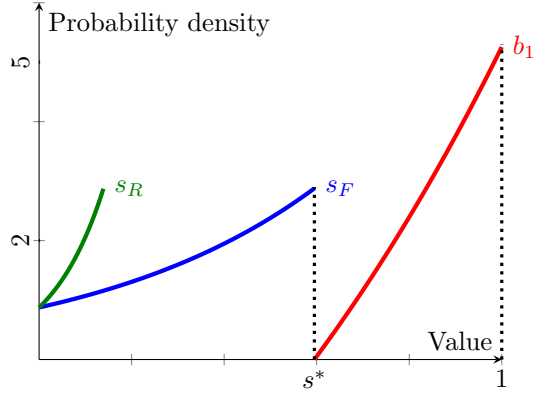
\begin{figure}[htbp]
\centering
\begin{tikzpicture}
    \begin{axis}[
	axis lines=middle,
        xlabel=Value, 
		ylabel =Probability density,
		clip=false,
		xticklabel=\empty,
        yticklabel=\empty,
        xmin=0, xmax=1, ymin=0, ymax=6,
        samples=100,
		width=.55\textwidth,
		height=.45\textwidth
    ]

    \addplot[domain=0.595:1, \thickness, red] {(2*(exp(x-0.595)-1))/(2*0.595-1)} node[at end, right] {$b_1$};
    \addplot[domain=0:0.595, \thickness, blue] {2*0.4378*exp(2*x)} node[at end, right] {$s_F$};
    \addplot[domain=0:0.139, \thickness, green!50!black] {2*0.4378*exp(2*(x/0.234))} node[at end, right] {$s_R$};
        
        \AxisTick{y}{5}{0}{5}
        \AxisTick{y}{2}{0}{2}
		\AxisTick{x}{0.595}{2.9}{$s^*$}
		\AxisTick{x}{1}{5.3}{1}
		
    \end{axis}
\end{tikzpicture}
\caption{Probability density of the independent parameters $s_F$ and $b_1$, and $s_R$ depending on $s_F$.}
    \label{fig:densityLBR}
\end{figure}

Both draws are made before the signal arrives, and resulting probability densities are illustrated on Figure~\ref{fig:densityLBR}. Once $s_F$ and $b_1$ are
fixed, the response is deterministic and is given by:
\begin{equation}
\label{eq:rand-upper-response}
f(s)=
\begin{cases}
\infty,&0\leq s<s_R,\\
b_1,&s_R\leq s<s_F,\\
s_F,&s\geq s_F.
\end{cases}
\end{equation}
Thus, after a very early signal the algorithm rents forever, after an
intermediate signal it buys at the drawn value of $b_1$, and if no signal arrives before
the drawn value of $s_F$, it buys at the fallback.

\begin{theorem} 
\label{th:rand-lower} 
Every randomized algorithm for the \USSR problem has competitive ratio at 
least $C_{\mathrm{rand}}$. 
\end{theorem} 
 
\begin{proof} 
The proof follows the same structure as in Theorem~\ref{th:det-lower}. For any deterministic response function $f$, let 
$$ 
C_f(t)=t\,\ecostf=\int_0^t\costf\,ds. 
$$ 
Multiplying by $t$ removes the density $1/t$ of the uniform signal. We 
 construct a stochastic adversary depicted by a random variable $T$ representing the time horizon $t$. The objective is to ensure that $\mathbb E[C_f(T)]\geq1/2$ 
for every $f$
while $\mathbb E[T\OPT(T)]$ is small enough to obtain the desired ratio. 
 
Set 
$$ 
s^*=\ln(2e/3) 
\qquad\text{and}\qquad 
\eta=\frac12e^{-2s^*}. 
$$

We define $T$ as follows. Start with an exponential random variable
$E_2\sim\mathrm{Exp}(2)$. If $E_2<s^*$, set $T=E_2$. Otherwise, the
horizon has reached the cutoff $s^*$: with probability $1/2$, stop
exactly at $s^*$, and with probability $1/2$, continue for an
independent exponential time $E_1\sim\mathrm{Exp}(1)$, that is, set
$T=s^*+E_1$. 
 
Equivalently, $T$ has density 
\begin{equation} 
\label{eq:rand-law} 
d_T(t)= 
\begin{cases} 
2e^{-2t},&0<t<s^*,\\[1mm] 
\eta e^{s^*-t},&t>s^*, 
\end{cases} 
\end{equation} 
together with a point probability 
$$ 
\mathbb P(T=s^*)=\eta. 
$$ 
 
The first part has total probability $1-e^{-2s^*}$, while the remaining
probability $e^{-2s^*}$ is split equally between the point mass at
$s^*$ and the exponential tail after $s^*$. Thus $T$ is well defined. 
 
We now explain the main idea of the proof. Fix a deterministic response function
$f$ with fallback $s_F$. Let $g$ be the function equal to $\infty$ on $[0,s_F]$
and to $s_F$ afterwards. As in Theorem~\ref{th:det-lower}, denoting
$
\Delta(x)
=
\mathbb E[(1+x-T)\mathds 1_{\{T\geq x\}}],
$ (with $\Delta(\infty)=0$)
we have for all $0<s<s_f$, $\Delta(f(s))=E_T[\costf-\costfst g s t]$ and
\begin{equation} 
\label{eq:rand-certificate} 
\mathbb E[C_f(T)] 
= 
E[C_g(T)]+\int_0^{s_F}\Delta(f(s))\,ds.
\end{equation} 

Therefore, to prove $\mathbb E[C_f(T)]\geq1/2$ for every deterministic
response, it is enough to show that our choice of $T$ satisfies
$$
E[C_g(T)]\geq\frac12
\qquad\text{for any }s_F,
$$
and
$$
\Delta(x)\geq0
\qquad\text{for every purchase time }x.
$$
We now verify these two properties.

Direct integration gives 
\begin{equation} 
\label{eq:rand-tails} 
\mathbb P(T\geq x)= 
\begin{cases} 
e^{-2x},&0\leq x\leq s^*,\\ 
\eta e^{s^*-x},&x>s^*, 
\end{cases} 
\end{equation} 
and 
\begin{equation} 
\label{eq:rand-tail-moment} 
\mathbb E[T\mathds 1_{\{T\geq x\}}]= 
\begin{cases} 
(x+\frac12)e^{-2x},&0\leq x\leq s^*,\\ 
(x+1)\eta e^{s^*-x},&x>s^*. 
\end{cases} 
\end{equation} 
The formulas at $x=s^*$ include the point probability at $s^*$. In
particular,
$$ 
\mathbb E[T]=\frac12. 
$$ 
 
\emph{The fallback term.}
For a fixed horizon $t$, we have:

$$
C_{g}(t) = 
\underbrace{ 
t 
\cdot 
\frac{\min(s_F,t)}{t} 
\cdot 
t 
}_{\substack{\text{signal arrives before the fallback:}\\ 
\text{probability } \min(s_F,t)/t,\; 
\text{cost } t,\; 
\text{rescaling factor } t}} \quad+
\underbrace{ 
t 
\cdot 
\frac{(t-s_F)\mathds 1_{\{t\geq s_F\}}}{t} 
\cdot 
(1+s_F) 
}_{\substack{\text{no signal before the fallback:}\\ 
\text{probability } (t-s_F)/t,\; 
\text{cost } 1+s_F,\; 
\text{rescaling factor } t}}. 
$$
 
Hence 
$$
C_{g}(t)
=
t\min(s_F,t)
+
(1+s_F)(t-s_F)\mathds 1_{\{t\geq s_F\}}.
$$
 
Averaging over the random horizon $T$ gives 
\begin{align*} 
R(s_F) := E[C_g(T)]
&= 
\mathbb E[T\min(s_F,T)] 
+(1+s_F)\mathbb E[(T-s_F)\mathds 1_{\{T\geq s_F\}}]\\ 
&= 
\int_0^{s_F} 
\mathbb E[T\mathds 1_{\{T\geq s\}}]\,ds 
+(1+s_F) 
\left( 
\mathbb E[T\mathds 1_{\{T\geq s_F\}}] 
-s_F\mathbb P(T\geq s_F) 
\right). 
\end{align*} 
 
By definition of $T$, this function is continuous except at $s^*$, and the following shows that it is continuous at $s^*$ as well:
\begin{align*}
\lim_{\varepsilon\to0^+}R(s^*+\varepsilon)-R(s^*-\varepsilon) &= (1+s^*)(E[T\mathds 1_{\{T=s^*\}}]-s^*\mathbb P(T=s^*))=0
\end{align*}

We also have $R(0)=\mathbb E[T]=\frac12$.

Away from $s^*$, the distribution of $T$ has a density $d_T(s_F)$, so 
$$
\frac{d}{ds_F}\mathbb P(T\geq s_F)=-d_T(s_F) 
\qquad\text{and}\qquad 
\frac{d}{ds_F}
\mathbb E[T\mathds 1_{\{T\geq s_F\}}]
=-s_F d_T(s_F). 
$$ 
Therefore, 
\begin{align*} 
R'(s_F) &= 
\mathbb E[T\mathds 1_{\{T\geq s_F\}}] 
+\left( 
\mathbb E[T\mathds 1_{\{T\geq s_F\}}] 
-s_F\mathbb P(T\geq s_F) 
\right)\\
&~~~~~~~~~~~~~~~~~~~+
(1+s_F) 
\left( 
-s_F d_T(s_F) 
-\mathbb P(T\geq s_F) 
+s_F d_T(s_F) 
\right)\\ 
&= 
2\mathbb E[T\mathds 1_{\{T\geq s_F\}}] 
-(1+2s_F)\mathbb P(T\geq s_F)\\ 
&= 
\begin{cases} 
2\left(s_F+\frac12\right)e^{-2s_F} 
-(1+2s_F)e^{-2s_F}, 
& s_F<s^*,\\[2mm] 
2(s_F+1)\eta e^{s^*-s_F} 
-(1+2s_F)\eta e^{s^*-s_F}, 
& s_F>s^*, 
\end{cases}\\ 
&= 
\begin{cases} 
0,&s_F<s^*,\\ 
\eta e^{s^*-s_F},&s_F>s^*. 
\end{cases} 
\end{align*} 
 
Thus $R$ is constant on $[0,s^*]$ and increasing afterwards. Hence
its minimum is attained for every $s_F\in[0,s^*]$, and in particular
$$
R(s_F)\geq R(0)=\frac12
$$
for every fallback $s_F$, including $s_F=\infty$ by taking limits.

\emph{Purchases after the signal.}
A quick computation gives
\begin{align*} 
\Delta(x) 
&=\mathbb E[(1+x-T)\mathds 1_{\{T\geq x\}}]\nonumber\\
&=(1+x)\mathbb P(T\geq x) 
-\mathbb E[T\mathds 1_{\{T\geq x\}}]\nonumber\\ 
&= 
\begin{cases} 
\frac12e^{-2x},&0\leq x\leq s^*,\\ 
0,&x>s^*. 
\end{cases} 
\end{align*} 
Therefore
$
\Delta(x)\geq0
$
for every purchase time.

Using the bounds on $\Delta$ and $R$ in~\eqref{eq:rand-certificate}, we obtain that every
deterministic response satisfies
$$
\mathbb E[C_f(T)]
=
R(s_F)+\int_0^{s_F}\Delta(f(s))\,ds
\geq
\frac12.
$$

It remains to compute the corresponding average of the optimal cost. 
Since $s^*<1$, we obtain 
\begin{align*} 
\mathbb E[T^2] 
&=2\int_0^\infty x\,\mathbb P(T>x)\,dx 
=\frac12+\eta,\\ 
\mathbb E[T\OPT(T)] 
&=\mathbb E[T^2] 
-\eta\int_1^\infty t(t-1)e^{s^*-t}\,dt\\ 
&=\frac12+\eta-3\eta e^{s^*-1}\\ 
&=\frac{1-e^{-2s^*}}2, 
\end{align*} 
where the last equality uses $3e^{s^*-1}=2$. 
 
Now consider any randomized algorithm with competitive ratio $C$. We describe it
as using a random variable $\Omega$ dictating which deterministic response
function $f_\Omega$ will be used.

For every fixed $t$, we have
$$ 
\mathbb E_\Omega[C_{f_\Omega}(t)] 
\leq C\cdot t\OPT(t). 
$$ 
Taking the expectation over the variable $T$, independently of the valuation of $\Omega$, and 
applying~\eqref{eq:rand-certificate} to each deterministic response 
$f_\Omega$, we get 
$$ 
C 
\geq 
\frac{\mathbb E_T\mathbb E_\Omega[C_{f_\Omega}(T)]}{\mathbb E(T\OPT(T))}
= 
\frac{\mathbb E_\Omega\mathbb E_T[C_{f_\Omega}(T)]}{\mathbb E(T\OPT(T))}
\geq\frac{1/2}{(1-e^{-2s^*})/2}
\geq \frac{4e^2}{4e^2-9} 
= 
C_{\mathrm{rand}}. 
$$ 
 
\end{proof}

\begin{remark}
We now provide some intuition on how the adversary strategy $T$ was determined.

The objective is to minimize simultaneously both terms of Equation
\eqref{eq:rand-certificate} for all possible fallback times and response functions:
\[
\mathbb E[C_f(T)]
=
 E[C_g(T)]+\int_0^{s_F}\Delta(f(s))\,ds.
\]
The first exponential part, of rate $2$, allows $E[C_g(T)]$ to be constant over
fallbacks $s_F\leq s^*$. Hence, under this adversary, all fallbacks in this
interval lead to the same expected cost for $g$ (the strategy renting before
$s_F$). Larger fallbacks will lead to larger costs.

The second exponential part, of rate $1$, allows us to ensure $\Delta(x)=0$ for
every purchase time $x>s^*$ (and stay positive for $x<s^*$), which implies $
E[C_f(T)]\geq  E[C_g(T)]$. In particular, assuming the signal belongs to
$[0,s_F]$, any purchase time larger than $s^*$ leads to the same cost, as at
any given time after $s^*$, the expected remaining renting cost equals the
buying cost.

The point probability at $s^*$ allows us to 
connect both regimes while preserving their properties.

Consequently, a large class of response functions attains the same value
$\mathbb E[C_f(T)]=1/2$: any fallback $s_F\leq s^*$, together with post-signal
purchases after $s^*$ (or renting forever), achieves equality. The
distribution of $T$ is therefore designed such that both the response function
shape and its fallback are irrelevant.

More generally, the same construction can be made with an arbitrary
cutoff $u\in(0,1)$ in place of $s^*$. The scaled expected cost remains equal to $1/2$ for this whole
class of response functions, and only the averaged optimal cost
$\mathbb E[T\OPT(T)]$ depends on $u$. Minimizing this denominator gives
$u=\ln(2e/3)=s^*$, which yields the strongest lower bound.
\end{remark}

\section{Stability to distribution deviation}
\label{sec:stab}

In this section, we assume that the true distribution of the signal is not uniform but follows a density function $f_D$, defined on $[0,1]$ then scaled to $[0,t]$. Let $d(f_D)=\sup_{x\in[0,1]}|f_D(x)-1|$ be the maximum density deviation between the assumed uniform density and the actual distribution.

\begin{theorem}
\label{th:deviation}
The optimal deterministic and randomized solutions to the \USSR problem achieve a competitive ratio larger by at most an additive $O(d(f_D))$ term when the signal distribution is drawn following the density function $f_D$ instead of the uniform function.
\end{theorem}

\begin{proof}

Let $F_D$ be the cumulative distribution function of the probability density $f_D$.
The signal density on $[0,t]$ is $f_D(s/t)/t$, whose deviation from the uniform density $1/t$ is at most $d(f_D)/t$. We write $\Delta$ for the absolute difference between the expected costs under the two signal distributions.
We first focus on the deterministic setting.

We assume the algorithm always follows the strategy depicted in Lemma~\ref{lem:UB}
and perform a case analysis based on the value of the adversarial time horizon $t$:
\begin{itemize}
\item $t<s_F$: the algorithm rents under both settings.
\item $t\in[s_F,s^*]$: with a probability $1-F_D(\frac{s_F}{t})$, the algorithm buys at $s_F$, and is otherwise optimal. The over-cost compared to the uniform density signal is then at most $1$ with a probability at most $|F_D(\frac{s_F}t)-\frac{s_F}t|\leq d(f_D)\cdot \frac{s_F}t$. As $t\geq s_F$ and $\OPT\geq t$, we obtain that the competitive ratio is then at most $d(f_D)/s_F\leq2.7d(f_D)$ higher.
\item $t\geq s^*$:  Recall that $h(t)\in[s_R,s_F]$ for $t\in[s^*,1].$ Extending $h(t)$ to be equal to $h(1)= s_R$ for $t>1$, we have:
$$  \ecostf  = \mathbb{P}(s < h(t)) \cdot t + \int_{h(t)}^{s_F} (f(s)+1) \frac{f_D(s/t)}{t}ds + \mathbb{P}(s \ge s_F) \cdot (1+s_F).$$
Using the density deviation bound and Lemma~\ref{lem:UB}, the expected cost difference when the signal is drawn from $f_D$ instead of uniformly is bounded by:
\begin{align*}
\Delta &\leq d(f_D)\cdot \frac{h(t)}{t}\cdot t + \frac{d(f_D)}{t}\cdot \int_{h(t)}^{s_F} (f(s)+1)ds + d(f_D)\cdot \frac{t-s_F}{t}\cdot (1+s_F)\\
&\leq C^*d(f_D)\OPT(t).
\end{align*}
The three terms in the first right-hand side are $d(f_D)$ times the expected cost under a uniform signal: its density is $1/t$ and its interval probabilities are $h(t)/t$ and $(t-s_F)/t$. Lemma~\ref{lem:UB} bounds this expected cost by $C^*\OPT(t)$, which gives the second inequality.
Therefore, the competitive ratio difference is at most $C^*d(f_D)\leq 1.5d(f_D)$.
\end{itemize}

\medskip

We now consider the randomized setting, so assume the algorithm follows the strategy described in Section~\ref{sec:rando}, and perform a similar case analysis:
\begin{itemize}
\item $t<s^*$: If the random parameter $s_F$ is larger than $t$, then the algorithm always rents. Otherwise, the overcost is at most $d(f_D)\cdot \frac{s_F}{t}$. Averaging over $s_F$, the expected cost difference for $t$ fixed is then at most:
\begin{align*}
\Delta &\leq \int_{0}^{t} 2(C_{rand}-1)e^{2x}\cdot d(f_D)\cdot \frac{x}{t} dx \\
&= \frac1{2t}(C_{rand}-1)d(f_D)[(2x-1)e^{2x}]_0^t\\
&=\frac1{2t}(C_{rand}-1)d(f_D)[(2t-1)e^{2t}+1].
\end{align*}
The optimal cost is $t$ and the displayed upper bound on $\Delta/t$ is increasing in $t$ over $(0,s^*]$ and bounded by $1.1d(f_D)$, so the competitive ratio difference is at most $1.1d(f_D)$.

\item $t\geq s^*$: We fix the values of $s_F$ (so $s_R$) and $b_1$, we then have:
$$\ecostf= \mathbb P (s<s_R)\cdot t+ \mathbb P (s\in[s_R,s_F)) \cdot(t\cdot\mathds 1_{t<b_1}+(1+b_1)\mathds{1}_{t\geq b_1}) +\mathbb P (s\geq s_F)\cdot (1+s_F).$$
The costs on $[s_R,t]$ are at most $2$, and the expected cost difference is bounded by:
$$\Delta \leq d(f_D)\cdot \frac{s_R}{t}\cdot t + d(f_D)\cdot \frac{t-s_R}t \cdot 2\leq d(f_D)\cdot (s^*+2)\leq 2.6 d(f_D).$$
This bound holds for every choice of the random parameters, and hence also after averaging over them.
So the competitive ratio difference is at most $2.6d(f_D)/\min\{1,t\}\leq 4.4d(f_D).$

\end{itemize}

\end{proof}

\section{Conclusion}
\label{sec:conclu}

In this paper, we have studied the classic ski rental problem under a novel
model allowing the algorithm to detect stochastic information correlated to an
adversarial input, which prevents Bayesian approaches as the input remains
chosen adversarially. Our tight bounds show that such information still
allows us to obtain significant improvements over the classic setting, even when
the algorithm is restricted to simplistic structures.

The main open direction is then to handle arbitrary signal distributions. We
have shown that our solutions are not sensitive to small perturbations in the
corresponding density functions (using the infinity norm), but the methods we
propose are tied to an underlying uniform distribution. The competitive ratio
gains will then heavily depend on the actual signal distribution, as
distributions highly concentrated on a small value are extremely valuable, while
distributions skewed towards the horizon time $t$ cannot help the algorithm.
Another direction is also to explore further problems where such a model is
relevant, apart from the already studied non-clairvoyant scheduling and
secretary problems.

\paragraph*{AI usage statement}
The authors used LLMs sparsely as a help to conduct the research reported in this manuscript, mostly to look for solutions of equations leading to the function depicted in Theorem~\ref{th:rand-upper}, and to perform faster various algebraic computations which were then checked independently.

\ifarxiv\bibliographystyle{plainnat}\fi
\bibliography{bibfile}

\newpage

\appendix

\section{Omitted proofs for Section~\ref{sec:step}}

We first state and prove several technical lemmas, before proving the optimal two and three-step response functions.

\subsection{Technical lemmas}

\begin{lemma}
\label{lem:stepincr}
For any 3-step response function starting by a renting step, there exists such a non-increasing response function with a non-larger competitive ratio.
\end{lemma}

\begin{proof}
By Lemma~\ref{lem:cr1.5}, we know that there exits a non-increasing 3-step response function starting by a renting step with competitive ratio smaller than $1.5$. Consider a 3-step response function $f$, starting by a renting step, with a competitive ratio smaller than $1.5$. $f$ equals infinity until $s_R$, then equals some value $b_1$ until $s_F$ and finally equals $s_F$. Assume by contradiction $b_1<s_F$.

First, we have $s_F<0.5$, as otherwise the expected cost
for large $t$ tends towards $1+s_F\geq1.5$.
The objective is to show that there exists a legal 3-step function $f'$, starting by a renting step, such that
$CR(f')\leq CR(f)$ and, for all $x<s'_F$, $f'(x)\geq s'_F$,
with $f'(x)=s'_F$ for $x\geq s'_F$.

Let $\Delta=s_F-s_R$ be the interval length on which
$f(x)<s_F$.
If $\Delta>s_F/2$, then the expected cost at $t=s_F$ is at least
$\frac{s_F+1}{2}\geq\frac32s_F$, as with probability at least half,
the algorithm buys before $t$ so pays at least 1, and pays at least
$s_F$ otherwise. This leads to a competitive ratio at least $1.5$,
so a contradiction. Therefore, $\Delta\leq s_F/2$.

Let $\delta=\Delta/2$. Consider $f_1$ equal to 
$\infty$ on $[0,s_F-\delta]$ and to $s_F$ on $[s_R+\delta,\infty)$.
We then build $f'$ from $f_1$ by decreasing its value for
$x\geq s_F-\delta$ from $s_F$ to
$s'_F:=s_F-\delta\geq3s_F/4$.

If $t<s_F-\delta$, we have $\eratioft{f'}{t}=1$.

Consider $t\geq s_F$, and recall $s_F<0.5$.
We express $\eratioft{f'}{t}$ in terms of $\eratioft{f}{t}$:
\begin{align*}
\eratioft{f'}{t}
&= \eratioft{f_1}{t}
   -\frac{(t-s_F+\delta)\delta}{t\min(1,t)}\\
&\leq \eratioft{f}{t}
   +\frac{-(t-s_F+\delta)\delta+\delta t+\delta s_F-\delta}
          {t\min(1,t)}\\
&= \eratioft{f}{t}
   +\frac{\delta(2s_F-1-\delta)}{t\min(1,t)}\\
&\leq \eratioft{f}{t}.
\end{align*}

For $s_F-\delta\leq t<s_F$, we aim to show that
$\eratioft{f'}t\leq\eratioft{f'}{s_F}$.
This inequality implies, as we have from the previous case
$\eratioft{f'}{s_F}\leq\eratioft{f}{s_F}$, that
$\eratioft{f'}{t}\leq CR(f)$.

We have:
\begin{align*}
\eratioft{f'}t
&= \frac1t\left(
     \frac{t-s'_F}{t}(1+s'_F)+\frac{s'_F}{t}\cdot t
   \right)\\
&= \frac{1+2s'_F}{t}-\frac{s'_F(1+s'_F)}{t^2}.
\end{align*}

Recalling $s_F<0.5$ and $\delta\leq s_F/4$, this expression
is maximized at:
\begin{align*}
t^*
&= \frac{2s'_F(1+s'_F)}{1+2s'_F}\\
&= s_F\cdot
   \frac{2(1-\frac{\delta}{s_F})
   (1+(1-\frac{\delta}{s_F})s_F)}
   {1+2(1-\frac{\delta}{s_F})s_F}\\
&\geq s_F\cdot
   \frac{2(1-\frac{\delta}{s_F})
   (1+(1-\frac{\delta}{s_F})\cdot0.5)}
   {1+2(1-\frac{\delta}{s_F})\cdot0.5}\\
&= s_F\cdot
   \frac{(1-\frac{\delta}{s_F})(3-\frac{\delta}{s_F})}
   {2-\frac{\delta}{s_F}}\\
&\geq s_F\cdot
   \frac{(1-\frac14)(3-\frac14)}{2-\frac14}\\
&= \frac{33}{28}s_F.
\end{align*}

Therefore, the displayed expression is increasing on
$[s_F-\delta,t^*]\supset[s_F-\delta,s_F]$.
Purchases triggered exactly at $t=s_F$ can only increase the cost
relative to its left-hand limit, so
\[
\eratioft{f'}t
\leq \lim_{u\uparrow s_F}\eratioft{f'}u
\leq \eratioft{f'}{s_F}
\leq CR(f).
\]
Together with the previous cases, this proves $CR(f')\leq CR(f)$.

\end{proof}

\begin{lemma}
\label{lem:buy01dstep}
For every $N\geq1$, the infimum competitive ratio among non-increasing
legal step response functions with at most $N$ finite steps and at most
one initial step at $+\infty$ is approached by response functions with
values in $[0,1]\cup\{+\infty\}$.
\end{lemma}

\begin{proof}
We show that every non-increasing legal step response function $f$
can be replaced by a function $f'$ with no more finite steps,
finite values at most 1, and $CR(f')\leq CR(f)$.

If $CR(f)\geq2$, the constant response function $f'=1$ suffices.
Hence assume $CR(f)<2$.
Since $CR(f)\geq1+s_F$, we have $s_F<1$.
Moreover, $s_F>0$, since fallback time zero gives an unbounded
competitive ratio.
If all finite values of $f$ are at most 1, there is nothing to prove.

Otherwise, let $b_1>1$ be the largest finite value of $f$, attained on
$[s_1,s_2)$, with $f=+\infty$ on $[0,s_1)$.
Let $b_2<b_1$ be the next finite value, which exists because the final value is $s_F<1$. We have
$
0\leq s_1<s_2\leq s_F<1$ and
$s_F\leq b_2<b_1$.

We construct a replacement that eliminates the value $b_1$.

We define:
\[
W:=s_2(b_1-s_F)+\int_{s_2}^{s_F}(f(u)-s_F)\,du\geq0.
\]
For $t\geq b_1>1$, all finite purchases have occurred and the offline
optimum is 1. Thus
\begin{align*}
\eratioft{f}{t}
&=s_1+\frac{
(s_2-s_1)(1+b_1)
+\int_{s_2}^{s_F}(1+f(u))\,du
+(t-s_F)(1+s_F)
}{t}\\
&=1+s_F+s_1+\frac{W-s_1(1+b_1)}{t}.
\end{align*}
In particular,
\[
\lim_{t\to\infty}\eratioft{f}{t}
=1+s_F+s_1\leq CR(f).
\]

\paragraph*{Case 1: $W\geq s_2(1+b_1)$.}

Define $f'=+\infty$ on $[0,s_2)$ and $f'=f$ elsewhere.
This removes the step at $b_1$, so in particular the term $(s_2-s_1)(1+b_1)$ from the expression of $\eratiof$.
For $t<b_1$, the algorithms coincide.
For $t\geq b_1$,
\[
\eratioft{f'}{t}
=1+s_F+s_2+\frac{W-s_2(1+b_1)}{t}.
\]
The numerator of the last term is nonnegative, so this expression is
non-increasing. Moreover,
\[
\eratioft{f'}{b_1}
=\eratioft{f}{b_1}-\frac{s_2-s_1}{b_1}
\leq\eratioft{f}{b_1}.
\]
Consequently for $t\geq b_1$,
\[
\eratioft{f'}{t}
\leq\eratioft{f'}{b_1}
\leq CR(f)
\]
and therefore $CR(f')\leq CR(f)$.

\paragraph*{Case 2: $W<s_2(1+b_1)$.}

We denote
$
s^*:=\frac{W}{1+b_1}\in[0,s_2)$ and
$b'_1:=\max\{1,b_2\}<b_1$, 
and define:
\[
f'(s)=
\begin{cases}
+\infty,&0\leq s<s^*,\\
b'_1,&s^*\leq s<s_2,\\
f(s),&s\geq s_2.
\end{cases}
\]
The function $f'$ is non-increasing because $b'_1\geq b_2$.
It is legal because $b'_1\geq1>s_F\geq s_2$, and its fallback
time remains $s_F$.

For $t<b'_1$, the algorithms coincide: both rent whenever the signal
lies in $[0,s_2)$, and their responses agree elsewhere.
For $t\geq b'_1\geq1$, all finite purchases of $f'$ have occurred, so
\begin{align*}
\eratioft{f'}{t}
&=1+s_F+s^*
+\frac{W-s_2(b_1-b'_1)-s^*(1+b'_1)}{t}\\
&=1+s_F+s^*
+\frac{(s^*-s_2)(b_1-b'_1)}{t}\\
&\leq1+s_F+s^*.
\end{align*}
Here the second equality uses $W=s^*(1+b_1)$.
Finally,
\[
1+s_F+s^*
=
\frac{
b_1\cdot\eratioft{f}{b_1}+1+s_F+s_1
}{1+b_1}
\leq \frac{b_1 \cdot CR(f)+CR(f)}{1+b_1}
\leq CR(f),
\]
because both $\eratioft{f}{b_1}$ and $1+s_F+s_1$ are at most $CR(f)$.
Thus $CR(f')\leq CR(f)$.

In Case 1, one finite step is deleted.
In Case 2, one infinite step may be introduced and the step at $b_1$ is replaced by a step at $b'_1$:
if $b'_1=b_2$, it merges with the next step; otherwise $b'_1=1$.
Hence neither case increases the number of finite steps, and each
eliminates one finite value larger than 1 without introducing a new one.
Repeating the construction proves the claimed domination statement.
In particular, restricting finite values to $[0,1]$ does not change
the infimum competitive ratio.

\end{proof}

\begin{lemma}
\label{lem:3stepb1eq1}
Consider a response function $f$ defined by the quantities $0\leq s_R\leq \tau \leq s_F\leq b_1\leq 1$ as follows: 
for $s \in [0,s_R)$, $f(s)=\infty$; for  $s \in [s_R,\tau)$, $f(s)=b_1$; for  $s \geq \tau$, $f(s)=s_F$. Then choosing $b_1:=1$ does not increase the competitive ratio.
\end{lemma}

\begin{proof}
Suppose $b_1<1$, and consider $f'$ equal to $f$ except that it equals $1$ instead of $b_1$ on $[s_R,\tau)$.

    For $t < b_1$, both function lead to the same behavior: renting.

    For $t \in [b_1,1)$, if $s\in[s_R,\tau)$, then $f$ leads to  a cost of $1+b_1$, while $f'$ leads to a smaller cost of  $ t < 1$.

    For $t \ge 1$:
    \begin{align}
        \eratiof &= \frac{s_R}{t} \cdot t + \frac{\tau-s_R}{t} \cdot (1+b_1) + \frac{t-\tau}{t} \cdot (1+s_F)  \nonumber \\
        &= 1+ s_R + s_F + \frac{b_1(\tau-s_R)-(s_F\tau+s_R)}{t} \label{eq:n3_b}
    \end{align}
    \begin{align}
        \eratioft{f'}{t} &= \frac{s_R}{t} \cdot t + \frac{\tau-s_R}{t} \cdot 2 + \frac{t-\tau}{t} \cdot (1+s_F)  \nonumber \\
        &= 1+ s_R + s_F + \frac{(\tau-s_R)-(s_F\tau+s_R)}{t} \label{eq:n3_1}
    \end{align}
    If $(\tau-s_R)-(s_F\tau+s_R)<0$, then $b_1(\tau-s_R)-(s_F\tau+s_R)<0$, so for both functions, $\eratioft{f}{\cdot}$ is increasing on $[1,\infty)$ and converges to $1+s_R+s_F$.

    Otherwise, $\eratioft{f'}{\cdot}$ is non-increasing over $[1,\infty)$ , so it attains its maximum at $t=1$:
    $$\eratioft{f'}{1} = (1+\tau+s_F) - (s_F\tau+s_R).$$

    We now show that this quantity is lower than $\eratioft{f}{b_1}$ so than $CR(f)$:
    \begin{align*}
        \eratioft{f}{b_1} &= \frac{1}{b_1} \left[ \frac{s_R}{b_1} \cdot b_1+ \frac{\tau-s_R}{b_1} \cdot (1+b_1) + \frac{b_1-\tau}{b_1} \cdot (1+s_F) \right]\nonumber \\
        &= \frac{1+\tau+s_F}{b_1} - \frac{\tau s_F+s_R}{b_1^2}.
    \end{align*}

	Then, we want to show:
    \begin{align*}
        \eratioft{f'}{1} &\le \eratioft{f}{b_1} \\
        \iff (1+\tau+s_F) - (s_F\tau+s_R) &\le \frac{1+\tau+s_F}{b_1} - \frac{\tau s_F+s_R}{b_1^2} \\
        \iff (\tau s_F+s_R) \left(\frac{1}{b_1^2} - 1\right) &\le (1+\tau+s_F) \left(\frac{1}{b_1} - 1\right) \\
        \iff (\tau s_F+s_R) \left(\frac{1}{b_1} + 1\right) &\le 1+\tau+s_F
    \end{align*}
    The last inequality follows from $\tau s_F+s_R \le \tau  - s_R \le \tau $ as $s_F\leq 1$, and $\tau \leq b_1$ so the left hand side is at most $\frac{\tau }{b_1} + \tau  \leq 1 + \tau $.

\end{proof}

\subsection{Optimal few-steps response functions}

We now determine the optimal strategy among the ones composed of exactly $2$ intervals on which $f$ is constant.

\begin{proof}[Proof of \Cref{lem:2step}]

Let $f_a^\tau(s_F)$ be the response function equal to $a$ for $s<\tau$ and to $s_F\in[\tau,1]$ otherwise. We want to show that for all $s_F\in[0,1]$, $a\in[s_F,+\infty]$ and $\tau\in[0,s_F]$, $CR(f_1^{\tilde s_F}(\tilde s_F)) \leq CR(f_a^\tau(s_F))$, for some value $\tilde s_F\approx 0.3353$. Note that we must have $a\geq \tau$ otherwise the response function is ill-defined.

Note that for all $a$, for a large adversarial time $t$, the algorithm associated to $f_a^\tau(s_F)$ pays at least a cost $1+s_F$, as the probability to buy at $s_F$ tends towards 1, $\lim_{t\rightarrow+\infty}\eratioft{f_a^\tau}{t} \geq 1+s_F$. As we will show that $CR(f_1^{\tilde s_F}(\tilde s_F))$  is smaller than $1.6$, we can assume $s_F\leq0.6$.

\begin{itemize}
\item We first show that $CR(f_1^{\tau}(s_F)) \leq CR(f_a^\tau(s_F))$ for all legal values of $s_F$, $a$, $\tau$.
\begin{itemize}
\item 
\textbf{Case 1:} $s_F \leq a < 1$.
 
By Lemma~\ref{lem:3stepb1eq1} with $s_R=0$, we conclude that $CR(f_1^\tau(s_F))\leq CR(f_a^\tau(s_F))$.

\item \textbf{Case 2:} $1 < a < \infty$.

For $t < 1$, both functions lead to the same cost $t$.

For $t \ge 1$, the competitive ratio associated to $f_1$ is:

\begin{align*}
\eratioft{f_1^\tau(s_F)}{t} &= \frac{\tau}{t} \cdot 2 + \frac{t-\tau}{t} \cdot (1+s_F) = 1+s_F + \frac{\tau(1-s_F)}{t}
\end{align*}

And its maximum is attained at $t=1$:

$$\max_{t\geq 1}\eratioft{f_1^\tau(s_F)}t =\eratioft{f_1^\tau(s_F)}1 = 1+s_F +\tau-\tau\cdot  s_F.$$

This maximal is lower than the competitive ratio obtained from $f_a^\tau$ at $t = a$:

\begin{align*}
\eratioft{f_a^\tau(s_F)}a &= \frac{\tau}{a} \cdot (1+a) + \frac{a-\tau}{a} \cdot (1+s_F) \\
&= 1+\tau+s_F-\frac{\tau\cdot s_F}{a} \ge 1+2s_F - s_F^2.
\end{align*}

Therefore, $CR(f_1^\tau(s_F))\leq CR(f_a^\tau(s_F))$.

\item \textbf{Case 3:} $a = +\infty$.

Again, both functions lead to the same cost if $t < 1$.

For $t \ge 1$:
\begin{align*}
\eratioft{f_a^\tau(s_F)}{t} &= \frac{\tau}{t} \cdot t + \frac{t-\tau}{t} \cdot(1+s_F) 
= 1+\tau+s_F - \frac{\tau(1+s_F)}{t}.
\end{align*}

which converges to $1+\tau+s_F > \max_{t \ge 1} \eratioft{f_1^\tau(s_F)}t = 1+\tau+s_F - \tau\cdot s_F$.
\end{itemize}

Therefore, in all cases, we obtain $CR(f_1^\tau(s_F))\leq CR(f_a^\tau(s_F))$.

\medskip

\item We now optimize the values of $s_F$ and $\tau$. Recall that we must ensure $\tau\leq s_F$, as otherwise the response function would not be valid.

 We first express $CR(f_1^\tau(s_F))$ and show that we can assume further $s_F=\tau$.

Recall that 
\begin{equation}
\label{lem:2stepCR1}
\max_{t\geq 1}\eratioft{f_1^\tau(s_F)}{t} =\eratioft{f_1^\tau(s_F)}{1} = 1+ \tau+s_F - \tau\cdot s_F.
\end{equation}

Note that this expression is increasing in $s_F$, so is minimized for $s_F=\tau$.

For $t \in [s_F,1)$, we have:
\begin{align}
\eratioft{f_1^\tau(s_F)}{t} &= \left( \frac{\tau}{t} \cdot t + \frac{t-\tau}{t} \cdot (1+s_F) \right) \cdot \frac{1}{t}  
=  \frac{1+\tau+s_F}{t} - \frac{\tau(1+s_F)}{t^2}. \label{eq:baselinee}
\end{align}

This expression in function of $t$ is increasing then decreasing and maximized for $t^* = \frac{2\tau(1+s_F)}{1+\tau+s_F}$ where it equals:
$$\max_{t<1}\eratioft{f_1^\tau(s_F)}{t} =\eratioft{f_1^\tau(s_F)}{t^*} = \frac{(1+\tau+s_F)^2}{4\tau(1+s_F)}. \qquad \text{if $t^*\in[s_F,1]$}$$

 This expression is increasing in $s_F$ for $\tau \in(0,1)$. Further, we have $\eratioft{f_1^\tau(s_F)}{s_F}$ and $\eratioft{f_1^\tau(s_F)}{1}$ both increasing in $s_F$. We know that, for $s_F$ and $\tau$ fixed, $\eratioft{f_1^\tau(s_F)}{t}$ is maximized over $t\in[s_F,1]$  at $t$ equal to either $s_F$, $t^*$ or $1$, depending on whether $t^*<s_F$ and whether $t^*>1$.
So $\max_{t<1} \eratioft{f_1^\tau(s_F)}{t}$ is increasing in $s_F$ so is minimized when $s_F=\tau $.

Therefore, the competitive ratio of $f_1^\tau(s_F)$ is minimized for $s_F=\tau $. 

Then, $t^*=\frac{2s_F+2s_F^2}{1+2s_F}$. As $s_F\leq0.6\leq \frac1{\sqrt2}$, we have $t^*\leq 1$. We also have $t^*>s_F$, therefore,
$$CR(f_1^\tau(s_F))=\max(\eratioft{f_1^\tau(s_F)}{1}, \eratioft{f_1^\tau(s_F)}{t^*}).$$

 We finally optimize the value of $s_F$.

 By the opposite monotonicity of $\eratioft{f_1^{s_F}(s_F)}{1}$ and $\eratioft{f_1^{s_F}(s_F)}{t^*}$ for $s_F\in[0,0.6]$, this leads to the following equations:
\begin{align*}
\eratioft{f_1^{s_F}(s_F)}{1} &= \eratioft{f_1^{s_F}(s_F)}{t^*}\\
    1 + 2s_F - s_F^2 &= \frac{(1+2s_F)^2}{4s_F(1+s_F)}\\
    1 + 2s_F - s_F^2 &= 1 + \frac{1}{4s_F(1+s_F)}\\
8 s_F^2 + 4 s_F^3 - 4 s_F^4 &= 1.
\end{align*}
The unique solution in $[0,1]$ is $s_F \approx 0.3353$ which leads to a competitive ratio $ \approx 1.5582 $.
\end{itemize}

\end{proof}

We now determine the best strategy among the ones composed of three intervals on which $f$ is constant, starting by $+\infty$.

\begin{proof}[Proof of \Cref{lem:3step}]

	Consider a response function $f$ of the \USSR problem minimizing the competitive ratio among the set of functions composed of at most three steps including a renting step.
	Let $s_R\leq \tau\leq s_F$ and $b_1$ such that:

    \begin{itemize}
        \item $s \in [0,s_R)$ : $f(s)=\infty$
        \item $s \in [s_R,\tau)$ : $f(s)=b_1$
        \item $s \geq \tau$ : $f(s)=s_F$
    \end{itemize}

	From Lemma~\ref{lem:stepincr}, we obtain that $b_1\geq s_F$ and $f$ is non-increasing. Then, from Lemma~\ref{lem:buy01dstep}, we know that $b_1\leq 1$. We therefore have $ 0 \le s_R \le \tau \le s_F \le b_1 \le 1$.
	
By Lemma~\ref{lem:3stepb1eq1}, we now assume further that $b_1=1$.

Fix the values of $\tau $ and $s_F$, and we aim at optimizing $s_R$. For $t<1$, the value of $s_R$ has no impact on $\eratiof$, as buying at time $1$ or $\infty$ both lead to renting until $t$. For $t\geq1$, we have
    \begin{align*}
        \eratiof &= \frac{s_R}{t} \cdot t + \frac{\tau -s_R}{t} \cdot 2 + \frac{t-\tau }{t} \cdot (1+s_F)\\
        & = 1+s_R+s_F + \frac{\tau (1-s_F)-2s_R}{t}.
    \end{align*}

We consider two cases:
\begin{itemize}
\item $s_R\geq \frac{\tau (1-s_F)}{2}$: then $CR(f)\geq\eratioft{f}{+\infty}= 1+s_R+s_F\geq1+s_F+\frac{\tau (1-s_F)}{2}$.
\item $s_R<\frac{\tau (1-s_F)}{2}$: then $CR(f)\geq\eratioft{f}{1}= 1+s_F+{\tau (1-s_F)}-s_R \geq 1+s_F+\frac{\tau (1-s_F)}{2}$.
\end{itemize}

Therefore, the optimal choice of $s_R$ is $\frac{\tau (1-s_F)}{2}$ which leads to a constant expected ratio $\eratioft{f}{t}$ for $t \ge 1 $.

	We now fix $\tau $ and compute $\eratiof$ for $t\in[s_F,1]$, in order to determine the best choice for $s_F$ (the argument is similar to the one in the proof of Lemma~\ref{lem:2step}):
\begin{align*}
\eratioft{f}{t} &= \frac1t\left(\frac {\tau } t \cdot t + \frac{t-\tau }{t}(1+s_F)\right) 
         = \frac{\tau +1+s_F}{t} - \frac{\tau (1+s_F)}{t^2}.
\end{align*}

This expression is increasing then decreasing and maximized at $t^*=\frac{2\tau (1+s_F)}{1+\tau +s_F}$ where $\eratioft{f}{t^*}=\frac{(\tau +1+s_F)^2}{4\tau (1+s_F)}$. This new expression is then increasing in $s_F$ for $\tau \in(0,1)$. Further, we have $\eratioft{f}{s_F}$ and $\eratioft{f}{1}$ both increasing in $s_F$. We know that, for $s_F$ and $\tau$ fixed, $\eratioft{f}{t}$ is maximized over $t\in[s_F,1]$  at $t$ equal to either $s_F$, $t^*$ or $1$.
So  $\max_{t<1} \eratioft{f}{t}$ is increasing in $s_F$ so is minimized when $s_F=\tau $ as we must have $s_F\geq\tau$. This is also true for $t\geq 1$, as $s_R=\frac{\tau (1-s_F)}{2}$. Therefore, the competitive ratio of $f$ is minimized for $s_F=\tau $. 

Therefore, for $t\geq 1$, we have $CR(f)=1+\frac{s_F(3-s_F)}{2}$. And for $t\in [s_F,1)$, the maximum is reached at $s^* = \frac{2s_F(1+s_F)}{1+2s_F}$. We indeed have $s^*\in[s_F,1]$ for $s_F\in[0,0.6]$. We have $\eratioft{f}{s^*}=1 + \frac{1}{4s_F(1+s_F)}$.

    The two functions have opposing monotonicities in $s_F$, so the maximum of them is minimized at their intersection:
    \begin{align*}
        \frac{1}{4s_F(1+s_F)} &= \frac{s_F(3-s_F)}{2}.
    \end{align*} 
	The unique solution in $[0,1]$ is $s_F \approx 0.3723$, which leads to $s_R \approx 0.1168$ and $CR(f) \approx 1.489$.
\end{proof}

\section{Omitted proofs for Section~\ref{sec:conti}}

Consider a response function $f$ equal to $\infty$ on an initial interval
$[0,s_R)$, strictly decreasing and continuous on $[s_R,s_F)$, and constant and
equal to $s_F$ for $s \geq s_F$. Let
$$
s^*:=\lim_{s\rightarrow s_F^-}f(s),
$$
and assume that $f$ maps $[s_R,s_F]$ onto $[s^*,1]$.

Let $h:[s^*,1]\rightarrow[s_R,s_F]$ be the inverse of $f$ restricted
to $[s_R,s_F]$. We define
\begin{equation}
\label{eq:det-I}
I(t)=\int_{h(t)}^{s_F}f(u)du \qquad \text{ and } \qquad I=I(1)=\int_{s_R}^{s_F}f(u)du.
\end{equation}

The expected cost associated to $f$ can then be expressed as follows.

\begin{itemize}
    \item $t \in [0, s_F)$: for any signal, the function leads to a buying time larger than $t$:
    \begin{equation}
        \label{eq:det-cost-1}
        \ecostf = t
    \end{equation}
    
    \item $t \in [s_F, s^*)$: a signal $[s_R,s_F)$ leads to a buying time larger than $s^*$, and later signals lead to a buying time at $s_F$.:
    \begin{align}
        \ecostf &= \mathbb{P}(s < s_F) \cdot t + \mathbb P(s \ge s_F) \cdot (1+s_F) \nonumber \\
        &= \frac{s_F}{t} \cdot t + \frac{t-s_F}{t} \cdot (1+s_F)  = 1 + 2s_F - \frac{s_F(1+s_F)}{t}
        \label{eq:det-cost-2}
    \end{align}
    
    \item For $t \in [s^*,1)$: a signal $s < h(t)$ leads to a cost of $t$, a signal $s \in [h(t),s_F)$  leads to buying at $f(s)$,  and later signals lead to buying at $s_F$: 
    \begin{align}
        \ecostf  &= \mathbb{P}(s < h(t)) \cdot t + \int_{h(t)}^{s_F} \frac{f(s)+1}{t} ds + \mathbb{P}(s \ge s_F) \cdot (1+s_F) \nonumber \\
        &= \frac{h(t)}{t} \cdot (t) + \frac{I(t)+s_F-h(t)}{t} + \frac{t-s_F}{t} \cdot (1+s_F) \nonumber \\
        &=  1 + s_F + h(t) + \frac{I(t) - h(t) - s_F^2}{t} 
        \label{eq:det-cost-3}
    \end{align}
    
    \item For $t \in [1, \infty)$: apart from the initial renting interval, all other buying times are activated:
    \begin{align}
       \ecostf  &= \mathbb{P}(s < s_R) \cdot t + \int_{s_R}^{s_F} \frac{f(s)+1}{t}  ds + \mathbb{P}(s \ge s_F) \cdot (1+s_F) \nonumber \\
        &= \frac{s_R}{t} \cdot (t) + \frac{I+s_F-s_R}{t} + \frac{t-s_F}{t} \cdot (1+s_F) \nonumber \\
        &=  1 + s_R + s_F + \frac{I - s_R - s_F^2}{t}
        \label{eq:det-cost-4}
    \end{align}  

\end{itemize}

\begin{lemma}
\label{lem:UB}
There exists a response function for the \USSR problem leading to a
competitive ratio of $C^*$.
\end{lemma}

\begin{proof}

Set
\begin{align*}
s^* &= \ln \left(\frac{2e}{3} \right) \approx0.595
&s_F &= \frac{\sqrt{1+{s^*}^2}-(1-s^*)}{2} \approx 0.379 \\
C^* &= \frac 12 + \frac{\sqrt{1+{s^*}^2}}{2s^*} \approx 1.478
&s_R &= C^*-1-s_F \approx0.099.
\end{align*}
These definitions imply
$$
2C^*s^*=1+2s_F,
\qquad
s^*=\frac{2s_F(1+s_F)}{1+2s_F},
\qquad
C^*=1+\frac{1}{4s_F(1+s_F)}.
$$

Define
$$
h(t)=-3C^*e^{t-1}+2C^*t+2C^*-(1+s_F)
$$
for $t\in[s^*,1]$. Since $e^{s^*-1}=2/3$,
$$
h(s^*)=s_F, \qquad
h(1)=s_R, \quad \text{ and } \quad
h'(t)=C^*(2-3e^{t-1})<0 \text{ for }t\in(s^*,1].
$$
Hence $h$ is continuous and strictly decreasing from $s_F$ to $s_R$.

We define $f=h^{-1}$ on $[s_R,s_F)$, and extend it by
$$
f(s)=+\infty \text{ for }s<s_R,
\qquad
f(s)=s_F \text{ for }s \geq s_F.
$$
Since $f(s)\in[s^*,1]$ on $[s_R,s_F)$ and $s^*>s_F$, this response is
legal.

For completeness, its finite decreasing part can be written explicitly as
$$
f(s)
=
\frac{s+1+s_F}{2C^*}-1
-
W_{-1}\left(
-\frac32
\exp\left(\frac{s+1+s_F}{2C^*}-2\right)
\right).
$$

We now show that $CR(f)=C^*$.

\begin{itemize}
\item For $t<s_F$, by~\eqref{eq:det-cost-1}, we have $\eratiof=1$.

\item For $t\in[s_F,s^*]$, Equation~\eqref{eq:det-cost-2} gives
$$
\eratiof
=
\frac{1+2s_F}{t}
-
\frac{s_F(1+s_F)}{t^2}.
$$
This function is maximized at
$t^*=\frac{2s_F(1+s_F)}{1+2s_F}=s^*,$
where it equals
$1+\frac{1}{4s_F(1+s_F)}=C^*.$

\item For $t\in[s^*,1]$,  let
$$
C_f(t)=t\,\ecostf.
$$
Using~\eqref{eq:det-cost-3},
$$
C_f(t)
=
t(1+s_F+h(t))+I(t)-h(t)-s_F^2.
$$
Since
$I(t)=\int_{h(t)}^{s_F}f(u)du$
and $f(h(t))=t$, we have
$
I'(t)=-t h'(t).
$

Therefore,
\begin{align*}
C_f'(t)
&=
1+s_F+h(t)+t h'(t)+I'(t)-h'(t)\\
&=
1+s_F+h(t)-h'(t)=
2C^*t.
\end{align*}
The previous case gives
$
C_f(s^*)=C^*(s^*)^2,
$
so
$
C_f(t)=C^*t^2
$
for every $t\in[s^*,1]$. Hence $\eratiof=C^*$ on this whole interval.

\item For $t>1$, last case gives
$
C_f(1)=C^*.
$
On the other hand, Equation~\eqref{eq:det-cost-4} gives at $t=1$
$$
C_f(1)
=
1+s_R+s_F+I-s_R-s_F^2.
$$
Since $C^*=1+s_R+s_F$, we obtain
$$
I=s_R+s_F^2.
$$
Consequently, by~\eqref{eq:det-cost-4}, for every $t\geq1$,
$$
\ecostf
=
1+s_R+s_F
=
C^*.
$$

\end{itemize}

Therefore $\eratiof\leq C^*$ for every $t$, with equality for all
$t\geq s^*$, and
$$
CR(f)=C^*.
$$
\end{proof}

\section{Omitted proofs of Section~\ref{sec:rando}}

We consider the randomized response function defined in Section~\ref{sec:rando}, and we will prove that, for every horizon $t>0$,
$$
\mathbb E_{s_F,b_1}[\ecostf]
=
C_{\mathrm{rand}}\OPT(t).
$$

We first show that this algorithm is legal, and compute a few expressions that will be useful for showing the theorem.

\begin{lemma}
\label{lem:rand-upper-parameters}
The  distributions of $b_1$ and $s_F$ are well-defined, and the response function described
in~\eqref{eq:rand-upper-response} is legal. Moreover, we have
$
\mathbb E[s_F]
=
C_{\mathrm{rand}}s^*-\frac12$,
$
\mathbb E[b_1]
=
\frac{(s^*)^2}{2s^*-1}-1$, 
$
\mathbb E[s_F(1+s_F)]
=
C_{\mathrm{rand}}(s^*)^2$,
and $
\mathbb E[s_F-s_R]
=
C_{\mathrm{rand}}(2s^*-1).
$
\end{lemma}

\begin{proof}
The fallback density integrates to one since
$$
\int_0^{s^*}2(C_{\mathrm{rand}}-1)e^{2x}dx
=
(C_{\mathrm{rand}}-1)(e^{2s^*}-1)
=
1.
$$

The density of $b_1$ also integrates to one. Using
$e^{1-s^*}=3/2$,
$$
\int_{s^*}^1
\frac{2(e^{x-s^*}-1)}{2s^*-1}dx
=
\frac{2(e^{1-s^*}-2+s^*)}{2s^*-1}
=
\frac{2(s^*-\frac12)}{2s^*-1}
=
1.
$$

The coefficient defining $s_R$ is approximately $0.234$, hence
$$
0\leq s_R\leq s_F\leq s^*\leq b_1.
$$
Therefore the response is well defined and every post-signal purchase
occurs after the signal.

For the first moment of the fallback,
$$
\mathbb E[s_F]
=
\int_0^{s^*}2(C_{\mathrm{rand}}-1)xe^{2x}dx
=
(C_{\mathrm{rand}}-1)
\left[
e^{2x}\left(x-\frac12\right)
\right]_0^{s^*}
=
C_{\mathrm{rand}}s^*-\frac12,
$$
where we used
$$
(C_{\mathrm{rand}}-1)e^{2s^*}=C_{\mathrm{rand}}.
$$

Similarly,
$$
\mathbb E[s_F(1+s_F)]
=
\int_0^{s^*}2(C_{\mathrm{rand}}-1)x(1+x)e^{2x}dx
=
(C_{\mathrm{rand}}-1)
\left[x^2e^{2x}\right]_0^{s^*}
=
C_{\mathrm{rand}}(s^*)^2.
$$

We also have,
$$
s_F-s_R
=
\frac{C_{\mathrm{rand}}(2s^*-1)}
     {C_{\mathrm{rand}}s^*-1/2}s_F.
$$
Therefore,
$$
\mathbb E[s_F-s_R]
=
\frac{C_{\mathrm{rand}}(2s^*-1)}
     {C_{\mathrm{rand}}s^*-1/2}
\mathbb E[s_F]
=
C_{\mathrm{rand}}(2s^*-1).
$$

Finally,
\begin{align*}
E[b_1] &= \int_{s^*}^1
\frac{2x(e^{x-s^*}-1)}{2s^*-1}dx
=
\frac 2{2s^*-1} \left[(x-1)e^{x-s^*}-\frac12x^2\right]_{s^*}^1\\
&=
\frac 2{2s^*-1} \left[\frac12-s^*+\frac12(s^*)^2\right]
= \frac{(s^*)^2}{2s^*-1}-1.
\end{align*}

\end{proof}

\begin{proof}[Proof of Theorem~\ref{th:rand-upper}]
Fix a horizon $t>0$. As in the lower-bound proof, it is convenient to
multiply the expected cost by $t$, which removes the density $1/t$ of
the uniform signal. We will prove that for all $t$,
$$
t\mathbb E_{s_F,b_1}[\ecostf]
=
C_{\mathrm{rand}}t\OPT(t).
$$

The main idea is to start from the response that rents for the entire
horizon, whose rescaled cost is $t^2$, and then add the two possible
changes introduced by our response: buying at the fallback $s_F$, and
buying at the middle buying time $b_1$.

Fix first $s_F$ and $b_1$. If $s_F\leq t$, then for signals arriving
after $s_F$, the algorithm buys at $s_F$ instead of renting until $t$.
The corresponding signal interval has length $t-s_F$, and on this
interval the cost changes from $t$ to $1+s_F$. Hence the fallback adds
$$
(t-s_F)(1+s_F-t)\mathds 1_{\{s_F\leq t\}}
$$
to the rescaled cost.

Similarly, if $b_1\leq t$, then signals in $[s_R,s_F)$ cause a purchase
at $b_1$ instead of renting until $t$. This signal interval has length
$s_F-s_R$, and the cost changes from $t$ to $1+b_1$. Hence the middle
purchase adds
$$
(s_F-s_R)(1+b_1-t)\mathds 1_{\{b_1\leq t\}}.
$$

The two signal intervals are disjoint. Therefore, for fixed $s_F$ and
$b_1$,
$$
t\ecostf
=
\underbrace{t^2}_{\text{renting  baseline}}
+
\underbrace{(t-s_F)(1+s_F-t)\mathds 1_{\{s_F\leq t\}}}_{\text{fallback correction}}
+
\underbrace{(s_F-s_R)(1+b_1-t)\mathds 1_{\{b_1\leq t\}}}_{\text{$b_1$-purchase correction}}.
$$

We now average these three terms over $s_F$ and $b_1$. The fallback term
is obtained by integrating over the density of $s_F$. For the
$b_1$-purchase term, independence of $s_F$ and $b_1$, together with
Lemma~\ref{lem:rand-upper-parameters}, gives
\begin{align*}
\mathbb E_{s_F,b_1}
\left[
(s_F-s_R)(1+b_1-t)\mathds 1_{\{b_1\leq t\}}
\right]
&=
\mathbb E[s_F-s_R]
\int_{s^*}^{\min(t,1)}
\frac{2(e^{x-s^*}-1)}{2s^*-1}
(1+x-t)dx\\
&=
2C_{\mathrm{rand}}
\int_{s^*}^{\min(t,1)}
(e^{x-s^*}-1)(1+x-t)dx.
\end{align*}

Consequently,
\begin{align}
\label{eq:rand-upper-cost}
t\mathbb E_{s_F,b_1}[\ecostf]
&=
t^2
+
\int_0^{\min(t,s^*)}
2(C_{\mathrm{rand}}-1)e^{2x}
(t-x)(1+x-t)dx
\nonumber\\
&+
2C_{\mathrm{rand}}
\int_{s^*}^{\min(t,1)}
(e^{x-s^*}-1)(1+x-t)dx.
\end{align}

The rest of the proof consists in evaluating this expression. There are
three regimes. Before $s^*$, only the fallback may be reached. Between
$s^*$ and $1$, all fallbacks can be reached and only part of the
support of $b_1$ can be reached. After $1$, both random purchase times
have their entire support before the end of the horizon.

\begin{itemize}
\item 
\emph{Case 1: $0<t\leq s^*$.}

The $b_1$-purchase term (i.e. the last term in~\eqref{eq:rand-upper-cost}) is zero. The
fallback correction is
$$
\int_0^t
2(C_{\mathrm{rand}}-1)e^{2x}
(t-x)(1+x-t)dx
=
(C_{\mathrm{rand}}-1)
\left[-(x-t)^2e^{2x}\right]_0^t
=
(C_{\mathrm{rand}}-1)t^2.
$$
Hence
$$
t\mathbb E_{s_F,b_1}[\ecostf]
=
t^2+(C_{\mathrm{rand}}-1)t^2
=
C_{\mathrm{rand}}t^2.
$$

\item \emph{Case 2: $s^*\leq t\leq1$.}
All possible fallback times satisfy $s_F\leq t$. Hence the rent term and
the fallback correction together give
$$
t^2+
\mathbb E[(t-s_F)(1+s_F-t)]
=
(1+2\mathbb E[s_F])t
-\mathbb E[s_F(1+s_F)].
$$
Using Lemma~\ref{lem:rand-upper-parameters},
$$
t^2+
\mathbb E[(t-s_F)(1+s_F-t)]
=
2C_{\mathrm{rand}}s^*t
-C_{\mathrm{rand}}(s^*)^2.
$$

The $b_1$-purchase correction is
\begin{align*}
2C_{\mathrm{rand}}
\int_{s^*}^t
(e^{x-s^*}-1)(1+x-t)dx
&=
2C_{\mathrm{rand}}
\left[
(x-t)e^{x-s^*}
-x-\frac{x^2}{2}+tx
\right]_{s^*}^t
=
C_{\mathrm{rand}}(t-s^*)^2.
\end{align*}

Therefore,
$$
t\mathbb E_{s_F,b_1}[\ecostf]
=
2C_{\mathrm{rand}}s^*t
-C_{\mathrm{rand}}(s^*)^2
+C_{\mathrm{rand}}(t-s^*)^2
=
C_{\mathrm{rand}}t^2.
$$

\item \emph{Case 3: $t\geq1$.}
As in the previous case, the rent term and the fallback correction give
$$
2C_{\mathrm{rand}}s^*t
-C_{\mathrm{rand}}(s^*)^2.
$$

Now, every possible middle buying time satisfies $b_1\leq t$, so, using Lemma~\ref{lem:rand-upper-parameters}, the
$b_1$-purchase correction becomes
\begin{align*}
\mathbb E_{s_F}[s_F-s_R](1+\mathbb E_{b_1}[b_1]-t) &= C_{\mathrm{rand}}(2s^*-1)\left(\frac{(s^*)^2}{2s^*-1}-t\right) \\
&= C_{\mathrm{rand}} (s^*)^2 - t\cdot  C_{\mathrm{rand}}(2s^*-1).
\end{align*}

Hence
$$
t\mathbb E_{s_F,b_1}[\ecostf]
=
2C_{\mathrm{rand}}s^*t
-C_{\mathrm{rand}}(s^*)^2
+C_{\mathrm{rand}}(s^*)^2
-C_{\mathrm{rand}}(2s^*-1)t
=
C_{\mathrm{rand}}t.
$$

\end{itemize}

We have thus shown that, for every $t>0$,
$$
t\mathbb E_{s_F,b_1}[\ecostf]
=C_{\mathrm{rand}}\cdot t \cdot \min(1,t) = C_{\mathrm{rand}}\cdot t \cdot \OPT(t),$$ 
which proves the theorem.
\end{proof}

\end{document}